\title{Hardness of embedding simplicial complexes in $\R^d$
}

\newcommand{\cmt}[1]{
}
\def\kamsymb{{\rm a}}
\def\itisymb{{\rm b}}
\def\ethsymb{{\rm c}}
\def\snfsymb{{\rm d}}

\author{
{\sc Ji\v{r}\'{\i} Matou\v{s}ek}$^{\kamsymb,\itisymb, \ethsymb}$
\and 
{\sc Martin Tancer}$^{\kamsymb,\itisymb,\ethsymb}$
\and
{\sc Uli Wagner}$^{\ethsymb,\, \snfsymb}$
}

\documentclass[11pt,a4paper]{article}
\usepackage{a4wide}

\usepackage{latexsym}
\usepackage{amssymb,amsmath,amsthm}
\usepackage{bbm}
\usepackage{epsfig,graphicx}
\usepackage{url}

\newtheorem{theorem}{Theorem}[section]  %
\newtheorem{obs}[theorem]{Observation}
\newtheorem{lemma}[theorem]{Lemma}

\newtheorem{corol}[theorem]{Corollary}

\newcommand{\heading}[1]{\vspace{1ex}\par\noindent{\bf #1}}

\long\def\onefigure#1#2{
\begin{figure*}[tbp]
\begin{center}
#1
\end{center}
\caption{#2}
\end{figure*}
} 

\newcommand{\lipefig}[2]  
{\onefigure{\mbox{\includegraphics{#1}}}{\label{f:#1} #2} }

\makeatletter
\newcommand{\ProofEndBox}{{\ifhmode\unskip\nobreak\hfil\penalty50 \else
          \leavevmode\fi\quad\vadjust{}\nobreak\hfill$\Box$
            \finalhyphendemerits=0 \par}}%
\makeatother

\newcommand{\proofend}{\ProofEndBox\smallskip}

\newcommand{\NP}{\textup{NP}}

\newcommand{\R}{\mathbbm{R}}
\newcommand{\Z}{\mathbbm{Z}}
\newcommand{\N}{\mathbbm{N}}

\newcommand\eps{\varepsilon}
\def\:{\colon}

\newcommand\bd{\partial}
\newcommand{\polyh}[1]{|#1|}
\newcommand{\alterdef}[1]{\left\{ \begin{array}{ll}
#1 \end{array}  \right. }

\newcommand{\vko}{\mathfrak{o}}

\DeclareMathOperator{\dist}{dist}

\DeclareMathOperator{\interior}{int}

\newcommand\EMBED[2]{{\rm EMBED$_{#1\to#2}$}}
\newcommand\CG{G}
\newcommand\TG{X}

\begin{document}
\maketitle
{\renewcommand\thefootnote{\kamsymb}
\footnotetext{Department of Applied Mathematics,
Charles University, Malostransk\'{e} n\'{a}m. 25,
118~00~~Praha~1,  Czech Republic}
}
{\renewcommand\thefootnote{\itisymb}
\footnotetext{Institute of Theoretical Computer Science (ITI),
Charles University, Malostransk\'{e} n\'{a}m. 25,
118~00~~Praha~1,  Czech Republic}
}
{\renewcommand\thefootnote{\ethsymb}
\footnotetext{Institute of  Theoretical Computer Science,
ETH Zurich, 8092 Zurich, Switzerland}
}
{\renewcommand\thefootnote{\snfsymb}
\footnotetext{Research supported by the Swiss National Science Foundation (SNF Project 200021-116741)}
}
\begin{abstract}
 Let \EMBED kd be the following algorithmic
problem: Given a finite simplicial complex $K$ of dimension
at most $k$,
does there exist a (piecewise linear) embedding of $K$ into $\R^d$?
Known results easily imply polynomiality of \EMBED k2
($k=1,2$; the case $k=1$, $d=2$ is graph planarity)
and of \EMBED k{2k} for all $k\ge 3$. 

We show that
the celebrated result of Novikov on the algorithmic
unsolvability of recognizing the $5$-sphere
implies that \EMBED{d}{d} and \EMBED{(d-1)}{d} are undecidable for 
each $d\ge 5$.
Our main result is \NP-hardness of \EMBED{2}{4} and, more generally,
of \EMBED{k}{d}
for all $k,d$ with $d\ge 4$ and $d\ge k\ge (2d-2)/3$.
These dimensions fall
outside the \emph{metastable range} of 
a theorem of Haefliger and Weber, which characterizes
embeddability using the \emph{deleted product obstruction}. Our reductions are based on examples, due to
Segal, Spie\.z, Freedman, Krushkal, Teichner, and Skopenkov,
showing that outside the metastable range
the deleted product obstruction is not sufficient to 
characterize embeddability.
\end{abstract}

\section{Introduction}\label{s:intro}

Does a given (finite) simplicial complex\footnote{
We assume that the reader is somewhat familiar with
basic notions of combinatorial
topology (introductory chapters of books like
\cite{Munkres,Hatcher,Mat-top} should provide a sufficient
background). Terminology and basic facts 
concerning simplicial complexes will be recalled
in Section~\ref{s:pl-prelim}.
Later on, in some of the proofs, we will need other, slightly more
advanced topological notions and results, which would take
too much space to
define properly. We hope that the main ideas can be
followed even when such things are skipped.}
 $K$ of dimension at most $k$
admit an embedding into $\R^d$? We consider the computational
complexity\footnote{For basic definitions from computational complexity,
such as polynomial-time algorithm,
NP-hardness, reduction, or 3-SAT the reader can refer
to any introductory textbook on algorithms.
Very recent textbooks were written 
 by Arora and Barak \cite{AroraBarak}
and by Goldreich \cite{Goldreich}.
We can also highly recommend Wigderson's essay
\cite{Wigderson}.
}
 of this question, regarding $k$ and $d$ as fixed integers. To our surprise, 
this question has apparently not been explicitly addressed before (with the exception of $k=1$, $d=2$
which is graph planarity), as far we could find.
Besides its intrinsic interest for the theory of computing, an algorithmic view of a 
classical subject such as embeddability may lead to new questions and also to a 
better understanding of known results. For example, computation complexity can be 
seen as a concrete ``measuring rod'' that allows one to compare
the ``relative strength'' of various embeddability criteria, respectively of 
examples showing the necessity of dimension restrictions in the
criteria. Moreover, hardness results provide concrete  
evidence that for a certain range of the parameters (outside the so-called 
metastable range), no simple structural characterization of embeddability (
such as Kuratowski's forbidden minor criterion for graph planarity) is to be 
expected.

For algorithmic embeddability problems, we consider
\emph{piecewise linear} (PL) embeddings. Let us remark that
there are at least two other natural notions of embeddings
of simplicial complexes in $\R^d$: 
\emph{linear embeddings} (also called \emph{geometric realizations}),
which are more restricted than PL embeddings,
and \emph{topological embeddings}, which give us more freedom
than PL embeddings. We will recall the definitions in
Section~\ref{s:pl-prelim}; here we quickly illustrate the differences
with a familiar example: embeddings of $1$-dimensional simplicial
 complexes, a.k.a. simple graphs, into $\R^2$.
For a topological embedding,  the image of each edge can
be an arbitrary (curved) Jordan arc, for a PL embedding it
has to be a polygonal arc (made of finitely many straight
segments), and for a linear embedding, it must be a single
straight segment. For this particular case ($k=1$, $d=2$),
all three notions happen to give the same
class of embeddable complexes, namely, all planar graphs
(by F\'ary's theorem). For higher dimensions there
are significant differences, though, which we also
discuss in Section~\ref{s:pl-prelim}.

Here we are interested mainly in embeddability 
in the topological sense (as
 opposed to linear embeddability,
which is a much more geometric problem and one with
a very different flavor), but since it seems problematic
to deal with arbitrary topological embeddings effectively,
we stick to PL embeddings, which can easily be
represented in a computer.

We thus introduce the decision problem
\EMBED kd, whose input is a simplicial complex
$K$ of dimension at most $k$, and where the output
should be YES or NO depending on whether $K$
admits a PL embedding into $\R^d$.

We assume $k\le d$, since a $k$-simplex cannot be 
embedded in $\R^{k-1}$.
For $d\ge 2k+1$ the problem becomes trivial, 
since it is well known that 
\emph{every} finite $k$-dimensional simplicial
complex embeds in $\R^{2k+1}$, even linearly
(this result goes back to Menger).
In all other cases, i.e., $k\le d\le 2k$, there
are both YES and NO instances;
for the NO instances one can use, e.g.,
examples of $k$-dimensional complexes not embeddable
in $\R^{2k}$ due to
Van Kampen \cite{vanKampen} and Flores \cite{Flores}.

Let us also note that the complexity of this
 problem is monotone in $k$
by definition, since an algorithm for \EMBED kd also
solves \EMBED {k'}d for all $k'\le k$.

\heading{Tractable cases. } It is well known that
\EMBED 12 (graph planarity) is linear-time solvable \cite{HopcroftTarjan:EfficientPlanarityTesting-1974}.
Based on planarity algorithms and on a characterization
of complexes embeddable in $\R^2$ due to Halin and Jung
\cite{HalinJung}, it is not hard to come up with
a linear-time decision algorithm for \EMBED22.
Since we do not know of a reference, we
outline such an algorithm in Appendix~\ref{s:22e}. 
Daniel Kr\'al' (personal communication) has independently 
devised a linear-time algorithm that actually produces an embedding.

There are many problems in computational topology
that are easy for low dimensions (say up to dimension $2$
or $3$) and become intractable from some dimension on
(say $4$ or $5$); we will mention some of them later.
For the embeddability problem, the situation is subtler,
since there are tractable cases in arbitrarily high dimensions,
namely, \EMBED k{2k} for every $k\ge 3$.

The algorithm is based on ideas of Van Kampen \cite{vanKampen},
which were made precise by Shapiro \cite{Shapiro-realiz2d} and
independently by Wu \cite{Wu-realiz2d}. 
Since we are
not aware of any treatment in an algorithmic context,
and since some of the published descriptions
have a small flaw (a sign error) and others use
a somewhat advanced language of algebraic topology,
we give a self-contained elementary
presentation of the algorithm (but
not a proof of correctness) in Appendix~\ref{s:vkamp}.

\heading{Hardness. }
According to a celebrated result of Novikov (\cite{volodin-al};
also see, e.g., \cite{nabutovsky-einstein} for an exposition),
the following problem is algorithmically unsolvable: 
Given a $d$-dimensional simplicial complex, $d\ge 5$, 
decide whether it is homeomorphic to $S^d$, the
$d$-dimensional sphere. By a simple reduction we
obtain the following result:

\begin{theorem}\label{t:novikov}
\EMBED{(d-1)}d \textup{(}and hence also \EMBED{d}{d}\textup{)} is algorithmically undecidable
for every $d\geq 5$.
\end{theorem}

This has an interesting consequence, which in some sense
strengthens results of Brehm and Sarkaria 
\cite{BrehmSarkaria}:

\begin{corol}\label{c:subdiv}
 For every computable (recursive)
function $f\:\N\to\N$ and for every 
$d\ge 5$ there exist $n$ and a finite $(d-1)$-dimensional
simplicial complex $K$ with $n$ simplices
that PL-embeds in $\R^d$ but such that
no subdivision of $K$ with 
at most $f(n)$ simplices embeds linearly in~$\R^d$.
\end{corol}

Our main result is hardness for cases where $d\geq 4$ and
$k$ is larger than roughly $\frac{2}{3} d$.

\begin{theorem}\label{t:4hard} 
\EMBED kd is \NP-hard for every
pair $(k,d)$ with $d\ge 4$ and $d\ge k\ge \frac{2d-2}{3}$.
\end{theorem}

We prove a special case of this theorem, \NP-hardness of \EMBED24,
in Section~\ref{s:hard24}; the proof is somewhat more intuitive than
for the general case and it contains most of the ideas.
All the remaining cases are proved in Section~\ref{s:nphard}.

Let us briefly mention where the 
dimension restriction $k\ge (2d-2)/3$ comes from.
There is a certain necessary condition for embeddability
of a simplicial complex into $\R^d$, called
the \emph{deleted product obstruction}.
A celebrated theorem of Haefliger and Weber, which is
a far-reaching generalization of the ideas of Van Kampen mentioned above,
asserts that this condition is also \emph{sufficient}
provided that $k\le \frac23d-1$ (these $k$ are said to
lie in the \emph{metastable range}). The condition
on $k$ in Theorem~\ref{t:4hard} is exactly that $k$ must be
outside of the metastable range 
(we refer to Appendix~\ref{s:ha-we} for a brief discussion
and references).

There are examples showing
that the restriction to the metastable range 
in the Haefliger--Weber theorem is indeed necessary,
in the sense that whenever $d\ge 3$ and $d\ge k>(2d-3)/3$,
there are $k$-dimensional complexes that cannot be embedded
into $\R^d$ but the deleted product obstruction fails to detect this.
We use constructions of this kind,  namely, examples due to
Segal and Spie{\.z} \cite{segal-spiez},
Freedman, Krushkal, and Teichner \cite{FKT}, and
 Segal, Skopenkov, and Spie{\.z} \cite{segal-skopenkov-spiez},
as the main ingredient in our  proof of Theorem~\ref{t:4hard}.

\newcommand\unk{{\bf?}}
\newcommand\nee{~~~}
\newcommand\ja{{$+$}}
\newcommand\und{{UND}}
\newcommand\NPh{{NPh}}
\newcommand\pol{{P}}
\begin{table*}
\begin{center}
\begin{tabular}{c|ccccccccccccc}
&\multicolumn{10}{|c}{$d=$} \\
$k=$&~2~&~3~&4    &5    &6   &7   &8   &9   &10    &~11   &~12  &~13  &~14   \\
\hline
1 & \pol&\ja  &\ja  &\ja  &\ja &\ja &\ja &\ja &\ja   &\ja  &\ja &\ja &\ja   \\
2 & \pol&\unk &\NPh &\ja  &\ja &\ja &\ja &\ja &\ja   &\ja  &\ja &\ja &\ja   \\
3 &\nee&\unk &\NPh &\NPh &\pol &\ja &\ja &\ja &\ja   &\ja  &\ja &\ja &\ja   \\
4 &\nee&\nee &\NPh &\und &\NPh&\NPh&\pol &\ja &\ja   &\ja  &\ja &\ja &\ja   \\
5 &\nee&\nee &\nee &\und &\und&\NPh&\NPh&\unk&\pol   &\ja  &\ja &\ja &\ja   \\
6 &\nee&\nee &\nee &\nee &\und&\und&\NPh&\NPh&\NPh   &\unk &\pol&\ja &\ja   \\
7 &\nee&\nee &\nee &\nee &\nee&\und&\und&\NPh&\NPh   &\NPh &\unk&\unk&\pol
\end{tabular}
\end{center}
\caption{\label{kdtable}
The complexity of \EMBED kd (\pol\ $=$ polynomial-time solvable,
\und\ $=$ algorithmically undecidable,
\NPh\ $=$ \NP-hard, \ja\ $=$ always embeddable,
\unk\ $=$ no result known).}
\end{table*}

\heading{Discussion. } The current complexity
status of \EMBED kd is summarized in Table~\ref{kdtable}.
In our opinion,
the most interesting currently
open cases are $(k,d)=(2,3)$ and $(3,3)$.

These are outside the metastable range,
and it took the longest to find an example showing
that they are not characterized by the deleted product obstruction;
see \cite{Go-Sko}. That example does not seem to lend itself easily to a hardness reduction, though.

A variation on the proof of our undecidability
result (Theorem~\ref{t:novikov}) shows that
both \EMBED23 and \EMBED33 are at least as hard
as the problem of recognizing the $3$-sphere  (that is,
given a simplicial complex, decide whether it
is homeomorphic to $S^3$). 
The latter problem is in NP \cite{Ivanov:ComputationalComplexityDecisionProblems3DimensionalTopology-2008,S3inNP},
but no hardness result seems to be known.

For the remaining questionmarks in the table (with $d\ge 9$), which
all lie in the metastable range, it seems that existing
tools of algebraic topology, such as Postnikov towers
and/or suitable spectral sequences, could lead at least to decision
algorithms, or even to polynomial-time
algorithms in some cases. Here the methods of
``constructive algebraic topology'' mentioned below,
which imply, e.g., the computability of higher homotopy
groups, should be relevant. However, as is discussed,
e.g., in \cite{RomeroRubioSergeraert}, computability
issues in this area are often subtle, even for
questions considered well understood in classical
algebraic topology.
We hope to clarify these things in a future work.

Our NP-hardness results are probably not the final
word on the computational complexity of the corresponding
embeddability problems; for example, some or all of these
might turn out to be undecidable.


\heading{Related work.} Among the most important computational problems in topology are the \emph{homeomorphism problem} for manifolds, and the \emph{equivalence problem} for knots.
The first one asks if two given manifolds $M_1$ and $M_2$ (given as simplicial complexes, say) are homeomorphic. The second one asks if two given knots, i.e., \textup{PL} embeddings $f,g\: S^1 \to \R^3$, are equivalent, i.e., 
if there is a PL homeomorphism $h\:\R^3\to\R^3$ such that $f=h\circ g$. 
An important special case of the latter in the \emph{knot triviality} problem: 
Is a given knot equivalent to the trivial knot (i.e., the standard geometric
circle placed in $\R^3$)?

There is a vast amount of literature on computational problems for $3$-manifolds and knots. For instance, it is 
algorithmically decidable whether a given $3$-manifold is homeomorphic to $S^3$ \cite{Rubinstein:AlgorithmRecongnizing3Sphere-1995,Thompson:ThinPositionrecognitionProblemS3-1994}, or whether a given polygonal knot in $\R^3$ is trivial \cite{Haken:TheorieNormalflaechen-1961}. 
Indeed, both problems have recently shown to lie in \NP\ \cite{Ivanov:ComputationalComplexityDecisionProblems3DimensionalTopology-2008,S3inNP},
\cite{hass-lagarias-pippenger}. The knot equivalence problem is also algorithmically decidable
\cite{Haken:TheorieNormalflaechen-1961,Hemion:ClassificationHomeomorphisms2Manifolds-1979,Matveev:ClassificationSufficientlyLarge3Manifolds}, but nothing seems to be known about its complexity status.
We refer the reader to the above-mentioned sources and to \cite{agol-hass-thurston} for further results, background and references.

In higher dimensions, all of these problems are undecidable. Markov~\cite{Markov:InsolubilityHomeomorphy-1958} showed that the homeomorphism problem 
for $d$-manifolds is algorithmically undecidable for every $d\geq 4$.
 For $d\geq 5$, this was strengthened by Novikov 
to the undecidability of recognizing $S^d$ (or any other fixed
$d$-manifold), as was mentioned above.
Nabutovsky and Weinberger \cite{nabu-weing-knot} showed that for $d\geq 5$,
 it is algorithmically undecidable whether a given PL embedding $f\:S^{d-2} \to
 \R^d$ is equivalent to the standard embedding 
(placing $S^{d-2}$ as the ``equator'' of the unit sphere $S^{d-1}$, say).
For further undecidability results, see, 
e.g., 
\cite{NabutovskyWeinberger:AlgorithmicAspectsHomeomorphismProblem-1999} and
the survey by Soare~\cite{Soare:ComputabilityDifferentialGeometry-2004}.

Another direction of algorithmic research in topology is the
computability of \emph{homotopy groups}. While the fundamental
group $\pi_1(X)$ is well-known to be uncomputable
\cite{Markov:InsolubilityHomeomorphy-1958}, all
higher homotopy groups of a given finite simply connected simplicial
(or CW) complex are computable (Brown \cite{Brown-comput-homotopy}).
There is also a $\#P$-hardness result of Anick 
\cite{Anick-homotopyhard} for the computation of higher homotopy,
but it involves CW complexes presented in a highly compact
manner, and thus it doesn't seem to have any direct consequences
for simplicial complexes.
More recently, there appeared
several works (Sch\"on \cite{Schoen-effectivetop},
Smith \cite{smith-mstructures},
 and Rubio, Sergeraert, Dousson, and Romero, e.g.
 \cite{RomeroRubioSergeraert}) 
aiming at
making methods of algebraic topology, such
as spectral sequences, ``constructive''; 
the last of these has also resulted
in an impressive software called KENZO. 

A different line of research 
relevant to the embedding problem concerns linkless embeddings of graphs. 
Most notably, results of Robertson, Seymour, and Thomas
\cite{RobertsonSeymourThomas:SachsLinklessEmbeddingConjecture-1995}
on linkless embeddings provide an interesting \emph{sufficient}
condition for embeddability of a $2$-dimensional complex
in $\R^3$, and they can thus be regarded as one of the known
few positive
results concerning \EMBED23. We will briefly discuss this
in Section~\ref{s:linkless}.

\section{Preliminaries on PL topology}\label{s:pl-prelim}

Here we review definitions and facts related to piecewise linear
(PL) embeddings. We begin with very standard things
but later on we discuss notions and results which we found
quite subtle (although they might be standard
for specialists), in an area where
 it is sometimes tempting to consider as ``obvious'' something
that is unknown or even false. Some more examples
and open problems, which are not strictly necessary
for the purposes of the present paper, but which helped
us to appreciate some of the subtleties of the embeddability
problem, are mentioned in Appendix~\ref{a:plvstop}.
For more information on PL topology, and for facts mentioned below
without proofs, we refer to Rourke and Sanderson
\cite{RourkeSanderson}, Bryant \cite{Bryant-handbook},  or
Buoncristiano \cite{fragmentsoftop}.

\heading{Simplicial complexes. } 
We formally regard a simplicial complex
as a geometric object, i.e.,
 a collection $K$ (finite in our case)
of closed simplices in some Euclidean space $\R^n$
such that if $\sigma\in K$ and $\sigma'$ is a face of $\sigma$,
then $\sigma'\in K$ as well, and if $\sigma,\tau\in K$, then
$\sigma\cap\tau\in K$, too. We write
$V(K)$ for the \emph{vertex set}, i.e.,
the set of all $0$-dimensional simplices of $K$,
and $\polyh K$ for 
the \emph{polyhedron} of $K$, i.e.,
the union of all simplices in~$K$. 
Often we do not strictly distinguish between a simplicial
complex and its polyhedron; for example, by an embedding
of $K$ in $\R^d$ we really mean an embedding of $\polyh K$
into $\R^d$.

The \emph{$k$-skeleton} of $K$ consists of all simplices
of $K$ of dimension at most $K$. A \emph{subcomplex} of $K$
is a subset $L\subseteq K$ that is a simplicial complex.
A simplicial complex $K'$ 
is a \emph{subdivision} of $K$ if $\polyh{K'}=\polyh{K}$ and
 each simplex of $K'$ is contained in some simplex of $K$.

Two simplicial complexes $K$ and $L$ are \emph{isomorphic}
if there is a face-preserving bijection $\varphi\:V(K)\to V(L)$ 
of the vertex sets (that is, $F\subseteq V(K)$ is the vertex set
of a simplex of $K$ iff $\varphi(F)$ is the vertex set
of a simplex of $L$). Isomorphic complexes have homeomorphic
polyhedra. Up to isomorphism, a simplicial complex $K$
can be described purely combinatorially, by specifying
which subsets of $V(K)$ form vertex sets of simplices of $K$.

We assume that the input to the embeddability problem is given
in this form, i.e., as an abstract finite set system.



\heading{Linear and PL mappings of simplicial complexes. }
A \emph{linear} mapping of a simplicial complex $K$ into $\R^d$
is a mapping $f\:\polyh K\to\R^d$ that is linear
on each simplex.  More explicitly, each point $x\in\polyh K$ is a convex
combination $t_0v_0+t_1v_1+\cdots+t_s v_s$,
where $\{v_0,v_1,\ldots,v_s\}$ is the vertex set of some
simplex $\sigma\in K$ and $t_0,\ldots,t_s$ are nonnegative
reals adding up to $1$. Then we have
$f(x)=t_0f(v_0)+t_1f(v_1)+\cdots+t_s f(v_s)$.

A \emph{PL mapping} of $K$ into $\R^d$ is a linear mapping of
some subdivision $K'$ of $K$ into $\R^d$.

\heading{Embeddings.}
A general \emph{topological embedding} of $K$ into $\R^d$
is any continuous mapping $f\:\polyh K\to\R^d$ that is a homeomorphism
of $\polyh K$ with $f(\polyh K)$. Since we only consider finite simplicial 
complexes, this is equivalent to requiring that $f$ be injective.

By contrast, for a \emph{PL embedding} we require additionally that $f$ be PL,
and for a \emph{linear embedding} we are even more restrictive and insist that 
$f$ be (simplexwise) linear.

\heading{PL embeddings versus linear embeddings.}
In contrast to planarity of graphs, linear and PL embeddability 
do not always coincide in higher dimensions. Brehm~\cite{Brehm:NonpolyhedralMoebiusStrip-1983} constructed a triangulation of the M\"obius
strip that does not admit a linear embedding into $\R^3$. Using methods from the theory of oriented matroids, Bokowski and
Guedes de Oliveira~\cite{BokowskiGuedesDeOliveira-GenerationOrientedMatoids-2000} showed that for any $g\geq 6$,
there is a triangulation of the orientable surface of genus $g$ that does not admit a linear embedding into $\R^3$.
In higher dimensions, Brehm and Sarkaria \cite{BrehmSarkaria}
showed that for every $k\geq 2$, 
and every $d$, $k+1\le d\le 2k$,
there is a $k$-dimensional simplicial complex $K$
that PL embeds into $\R^{d}$ but does not admit a linear embedding. Moreover, for any given $r\geq 0$, there is such a $K$ such that even the $r$-fold barycentric subdivision $K^{(r)}$ is not linearly embeddable into~$\R^d$. Our Corollary~\ref{c:subdiv}
is another result of this kind.

On the algorithmic side, the problem of \emph{linear} embeddability of a given finite simplicial complex into $\R^d$ is at least
algorithmically decidable, and for $k$ and $d$ fixed,
it even belongs to \textup{PSPACE}
(since the problem can easily be formulated as the 
solvability over the reals of a system of polynomial inequalities with
integer coefficients, which lies in \textup{PSPACE}
\cite{Renegar:ComplexityFirstOrderTheoryReals-all3}).

\heading{PL structures. } Two simplicial complexes $K$ and $L$
are \emph{PL homeomorphic} if there are a subdivision $K'$ of $K$
and a subdivision $L'$ of $L$ such that $K'$ and $L'$ are
isomorphic.

Let $\Delta^d$ denote the simplicial complex consisting of
all faces of a $d$-dimensional simplex (including the simplex
itself), and let $\bd\Delta^d$ consist
of all faces of $\Delta^d$ of dimension at most $d{-}1$.
Thus, $\polyh{\Delta^d}$ is topologically $B^d$, the $d$-dimensional
ball, and $\polyh{\bd\Delta^d}$ is topologically $S^{d-1}$.

A $d$-dimensional \emph{PL ball} is a simplicial complex 
PL homeomorphic to $\Delta^d$, and a $d$-dimensional
\emph{PL sphere} is a simplicial complex PL homeomorphic to $\bd\Delta^{d+1}$.
Let us mention that a (finite) simplicial complex $K$ is PL embeddable
in $\R^d$ iff it is PL homeomorphic to a subcomplex of 
a $d$-dimensional PL ball (and similarly, $K$
is PL embeddable in $\polyh{\bd\Delta^{d+1}}$ iff
it is PL homeomorphic to a subcomplex of a $d$-dimensional
PL sphere).

One of the great surprises in higher-dimensional topology
was the discovery that simplicial complexes with homeomorphic
polyhedra need not be PL homeomorphic (the failure of
the ``Hauptvermutung''). In particular, there exist
\emph{non-PL spheres}, i.e., simplicial complexes
homeomorphic to a sphere that fail to be PL spheres.
More precisely, every simplicial complex homeomorphic
to $S^1$, $S^2$, $S^3$, and $S^4$ is a PL sphere,\footnote{The proof 
for $S^4$ relies on the recent solution of the Poincar\'e conjecture
by Perelman.} but there
are examples of non-PL spheres of dimensions $5$ and higher
(e.g., the double suspension of the 
\emph{Poincar\'e homology $3$-sphere}).


\heading{A weak PL Schoenflies theorem. }
The well-known Jordan curve theorem states that if  $S^1$
is embedded (topologically) in $\R^2$, the complement of the image
has exactly two components. Equivalently, but slightly
more conveniently, if  $S^1$ is embedded in $S^2$,
the complement has two components. The \emph{Schoenflies theorem}
asserts that in the latter setting, the closure of each of the components
is homeomorphic to the disk $B^2$.

While the Jordan curve theorem generalizes to an arbitrary dimension
(if $S^{d-1}$ is topologically embedded in $S^{d}$, the 
complement has exactly two components), the Schoenflies theorem
does not. There are embeddings $h\:S^2\to S^3$ such that
the closure of one of the components of $S^3\setminus h(S^2)$
is not a ball; a well known example is the 
\emph{Alexander horned sphere}.

The Alexander horned sphere is an infinitary construction;
one needs to grow infinitely many ``horns'' from
the embedded $S^2$ to make the example work.
In higher dimensions, there are strictly finite examples,
e.g., a $5$-dimensional subcomplex $K$ of a $6$-dimensional
PL sphere $S$ such that $\polyh K$ is topologically an $S^5$
(and $K$ is a non-PL sphere), but the closure of
a component of $\polyh S\setminus \polyh K$ is not
a topological ball (see Curtis and Zeeman \cite{curtis-zeeman-schoenflies}).

Thus, one needs to put some additional conditions on
the embedding to make a ``higher-dimensional Schoenflies theorem''
work. We will need the following version, in which we
assume a $(d-1)$-dimensional PL sphere sitting in
a $d$-dimensional PL sphere.
  
\begin{theorem}[Weak \textup{PL} Schoenflies Theorem]\label{t:weakPLScho}
Let $f$ be a \textup{PL} embedding of $\partial \Delta^d$ into 
$\partial \Delta^{d+1}$.
Then the complement $\polyh{\partial \Delta^{d+1}} \setminus 
f(\polyh{\partial \Delta^d})$ 
has two components, whose closures are topological $d$-balls.
\end{theorem}
 
For a proof of this theorem, see, e.g., 
\cite{Newman:DivisionEuclideanSpaceSpheres-1960} or \cite{Glaser:Schoenflies}. 
A simple, inductive proof is to appear in the upcoming revised edition 
of the book \cite{fragmentsoftop} by Buoncristiano and Rourke.

Let us remark that a ``strong'' PL Schoenflies theorem would claim
that under the conditions of Theorem~\ref{t:weakPLScho},
the closure of each of the components is a PL ball,
but the validity of this stronger statement is known
only for $d\le 3$, while for each $d\ge 4$ it is
(to our knowledge) an open problem.

\heading{Genericity. } First let us consider a linear mapping
$f$ of a simplicial complex $K$ into $\R^d$.
We say that $f$ is \emph{generic} if $f(V(K))$ 
is a set of distinct points in $\R^d$
in general position.

If $\sigma,\tau\in K$ are
disjoint simplices, then the intersection 
 $f(\sigma)\cap f(\tau)$  is empty for $\dim\sigma+\dim\tau<d$
and it has at most one point for $\dim\sigma+\dim\tau=d$.

A PL mapping of $K$ into $\R^d$ is generic if the corresponding linear
mapping of the subdivision $K'$ of $K$ is generic.

A PL embedding  can always
be made generic (by an arbitrarily small perturbation).

\heading{Linking and linking numbers. } 
Let $k,\ell$ be integers, and let $f\:S^k\to\R^{k+\ell+1}$
and $g\:S^\ell\to \R^{k+\ell+1}$ be PL embeddings with
$f(S^k)\cap g(S^\ell)=\emptyset$ (so here we regard $S^k$
and $S^\ell$ as PL spheres).
We will need two notions capturing how the images of $f$
and $g$ are ``linked'' (the basic example is $k=\ell=1$,
where we deal with two disjoint simple closed curves in $\R^3$).
For our purposes, we may assume that $f$ and $g$ are 
mutually generic (i.e. $f\sqcup g$, regarded
as a PL embedding of the disjoint union 
$S^k\sqcup S^\ell$ into $\R^{k+\ell+1}$, is generic).

The images $f(S^k)$ and $g(S^\ell)$ are \emph{unlinked} if
$f$ can be extended to a PL mapping $\bar f\:B^{k+1}\to \R^{k+\ell+1}$
of the $(k+1)$-dimensional ball such that $\bar f(B^{k+1})\cap
g(S^\ell)=\emptyset$. 

To define the \emph{modulo $2$ linking
number} of $f(S^k)$ and $g(S^\ell)$, we again extend
$f$ to a PL mapping $\bar f\:B^{k+1}\to \R^{k+\ell+1}$
so that $\bar f$ and $g$ are still mutually generic
(but otherwise arbitrarily).
Then the modulo $2$ linking number is the number of
intersections between $\bar f(B^{k+1})$ and
$g(S^\ell)$ modulo $2$ (it turns out that it does not
depend on the choice of $\bar f$). In the sequel,
we will use the phrase ``odd linking number''
instead of the more cumbersome ``nonzero linking number modulo $2$''
(although ``linking number'' in itself has not been properly defined).

These geometric definitions are quite intuitive.
However, alternative (equivalent in our setting
but more generally applicable) definitions
are often used, phrased in terms of homology
or mapping degree, which are in some respects easier to work
with (e.g., they show that linking is symmetric, i.e.,
$f(S^k)$ and $g(S^\ell)$ are unlinked iff
$g(S^\ell)$ and $f(S^k)$ are unlinked).

\section{Undecidability: Proof of Theorem~\ref{t:novikov}}
\label{sec:undecidability}

We begin with a statement of Novikov's result mentioned
in the introduction (undecidability
of $S^d$ recognition for $d\ge 5$) in a form convenient
for our purposes. 

\begin{theorem}[\textbf{Novikov}]\label{t:novik-cite}
 Fix $d\geq 5$. There is an effectively 
constructible sequence of simplicial complexes
$\Sigma_i$, $i\in \N$, with the following properties:
\begin{enumerate}
\item[\rm (1)] Each $\polyh{\Sigma_i}$ is a homology $d$-sphere.
\item[\rm (2)] For each $i$, either $\Sigma_i$ is a \textup{PL} $d$-sphere, or the fundamental group of $\Sigma_i$ is nontrivial (in particular, $\Sigma_i$ is not homeomorphic to the $d$-sphere).
\item[\rm (3)] There is no algorithm that decides for every given $\Sigma_i$ 
which of the two cases holds.
\end{enumerate}
\end{theorem}

We refer to the appendix in \cite{nabutovsky-einstein} for a detailed proof. 

We begin the proof of Theorem~\ref{t:novikov} with the following simple
lemma.

\begin{lemma}
\label{lem:skeleton-homology-sphere-minus-facet} 
Let $\Sigma$ be a simplicial complex whose
polyhedron is a homology $d$-sphere, $d\ge 2$.
(The same proof works for any homology $d$-manifold.)
Let $K$ be the $(d-1)$-skeleton of $\Sigma$.
For every $d$-simplex $\sigma \in \Sigma$, the
set $\polyh{K} \setminus \partial\sigma$ 
is path connected (here $\partial\sigma$ is the relative boundary
of $\sigma$). 
\end{lemma}

\begin{proof} By Lefschetz duality 
(see, e.g., \cite[Theorem~70.2]{Munkres}), $\polyh \Sigma
\setminus \sigma$ is path
 connected. Indeed, 
Lefschetz duality yields $H^0(\polyh \Sigma \setminus \sigma)\cong 
H_d(\Sigma, \sigma)$ 
(homology with $\Z_2$ coefficients, say). 
The exact homology sequence of the pair $(\Sigma,\sigma)$, 
together with the fact that $\sigma$ is contractible, 
yields $H_d(\Sigma)\cong H_d(\Sigma,\sigma)\cong \Z_2$.

Next, we claim that if $\gamma$ is a path in $\polyh{\Sigma}\setminus \sigma$ 
connecting two points $x,y \in \polyh K$,
then $x$ and $y$ can also be connected by
 a path in $\polyh K\setminus \bd\sigma$.
Indeed, given a $d$-dimensional simplex $\tau\in \Sigma\setminus\sigma$,
we have $\bd\tau\setminus\sigma$ path-connected. Hence
we can modify $\gamma$ as follows:
 Letting $a:=\min\{t : \gamma(t) \in \tau\}$ and 
$b:=\max\{t:\gamma(t)\in \tau\}$,
we replace the segment of $\gamma$ between
 $\gamma(a)$ and $\gamma(b)$ by a path $\eta$ in 
$\partial \tau \setminus  \sigma$.
Having performed this modification
for every $\tau\in \Sigma\setminus\sigma$
(in some arbitrary order), we end up with a path
 connecting $x$ and $y$ that lies entirely within 
$\polyh K \setminus \partial \sigma$.
\end{proof}

\begin{lemma}\label{l:emb-recogS}
Let $d\geq 2$. Suppose that $\Sigma$ is 
a homology $d$-sphere, and let $K$ be its $(d-1)$-skeleton. 
\begin{enumerate}
\item[\rm(i)]
If $\Sigma$ is a \textup{PL} sphere,
 then $K$ \textup{PL} embeds into $\R^d$.
\item[\rm(ii)] 
If $K$ \textup{PL} embeds into $\R^d$,
 then $\Sigma$ is homeomorphic to $S^d$.
\end{enumerate}
\end{lemma}

\begin{proof}
Part (i) is clear. 

For part (ii), let us
suppose that $f$ is a PL embedding of $K$ into $\R^d$. 
Since $K$ is compact, the image of $f$ is contained in some big 
$d$-dimensional simplex, and by taking this simplex as one facet of 
$\Delta^{d+1}$, we can consider $f$ as a PL embedding of $K$ 
into $\partial \Delta^{d+1}$. Consider a $d$-simplex $\sigma$ 
of $\Sigma$. By the weak PL Schoenflies theorem (Theorem~\ref{t:weakPLScho}), 
$\polyh{\partial \Delta^{d+1}}\setminus f(\partial \sigma)$ 
has two components, whose
 closures are topological $d$-balls. Moreover, since 
$\polyh K\setminus \partial \sigma$ is path-connected, 
its image under $f$ must be entirely contained in one of these
 components. 

Therefore, we can use the closure of the other 
component to extend $f$ to a topological embedding of $\sigma$. By
 applying this reasoning to each $d$-face, we obtain 
a topological embedding $g$ of $\Sigma$ into $\partial \Delta^{d+1}$. 
It follows for instance from  Alexander duality 
(see, e.g., \cite[Theorem~74.1]{Munkres}) that $g$ must be surjective, i.e., a homeomorphism.
\end{proof}

\heading{Proof of Theorem~\ref{t:novikov}. }
The undecidability of  \EMBED{(d-1)}{d} for $d\ge 5$ is 
an immediate consequence of Theorem~\ref{t:novik-cite}
and Lemma~\ref{l:emb-recogS}.
\proofend 

\heading{Proof of Corollary~\ref{c:subdiv}. }
Let us suppose that there is a recursive function
$f$ contradicting the statement. That is, every
$(d-1)$-dimensional
$K$ with $n$ simplices that PL-embeds in $\R^d$
at all has a subdivision with at most $f(n)$
simplices that embeds linearly. Then,
given a $(d-1)$-dimensional complex $K$ with $n$ simplices,
we could generate all subdivisions $K'$ of $K$
with at most $f(n)$ simplices (see Acquistapace
 et al.~\cite{acquistapace-al}, Proposition~2.15)
and, using the PSPACE algorithms
mentioned in Section~\ref{s:pl-prelim},
test the linear embeddability of each $K'$
in $\R^d$. This would yield a decision algorithm
for \EMBED{(d-1)}{d}, contradicting
Theorem~\ref{t:novikov}.
\proofend

\section{Hardness of embedding \boldmath $2$-dimensional
complexes in $\R^4$}\label{s:hard24}

We will reduce the problem $3$-SAT 
to \EMBED 24. 
Given a $3$-CNF formula $\varphi$,  we construct
a $2$-dimensional simplicial complex $K$ that
is PL embeddable in $\R^4$ exactly if $\varphi$ is satisfiable.

First we define two particular $2$-dimensional simplicial
complexes $\CG$ (the \emph{clause gadget}) and $\TG$
(the \emph{conflict gadget}).
They are closely related to the main example of Freedman et al.~\cite{FKT}:
$\TG$ is taken over exactly, and $\CG$ is a variation on
a construction in \cite{FKT} (which, in turn, is similar
in some respects to an example of Segal and Spie\.z \cite{segal-spiez},
with some of the ideas going back to Van Kampen \cite{vanKampen}).
 
\subsection{The clause gadget}\label{ss:clgadget}

To construct $\CG$, we begin with a $6$-dimensional
simplex on the vertex set $\{v_0,v_1,\ldots,v_6\}$, and
we let $F$ be the $2$-skeleton of this simplex
(F for ``full'' skeleton).
Then we make a hole in the interior of the three 
triangles ($2$-simplices) $v_0v_1v_2$,
$v_0v_1v_3$, and $v_0v_2v_3$.
That is, we subdivide each of the triangles and
from each of these subdivisions we remove
 a small triangle in the middle,
as is indicated in Fig.~\ref{f:e4h-C}.\footnote{Alternatively,
we could also make the clause gadget by simply
removing the the triangles $v_0v_1v_2$,
$v_0v_1v_3$, and $v_0v_2v_3$ from $F$.
However, the embedding of the resulting complex
$K$ for satisfiable formulas $\varphi$ would become somewhat
more complicated.}
This yields the simplicial complex $\CG$.
\lipefig{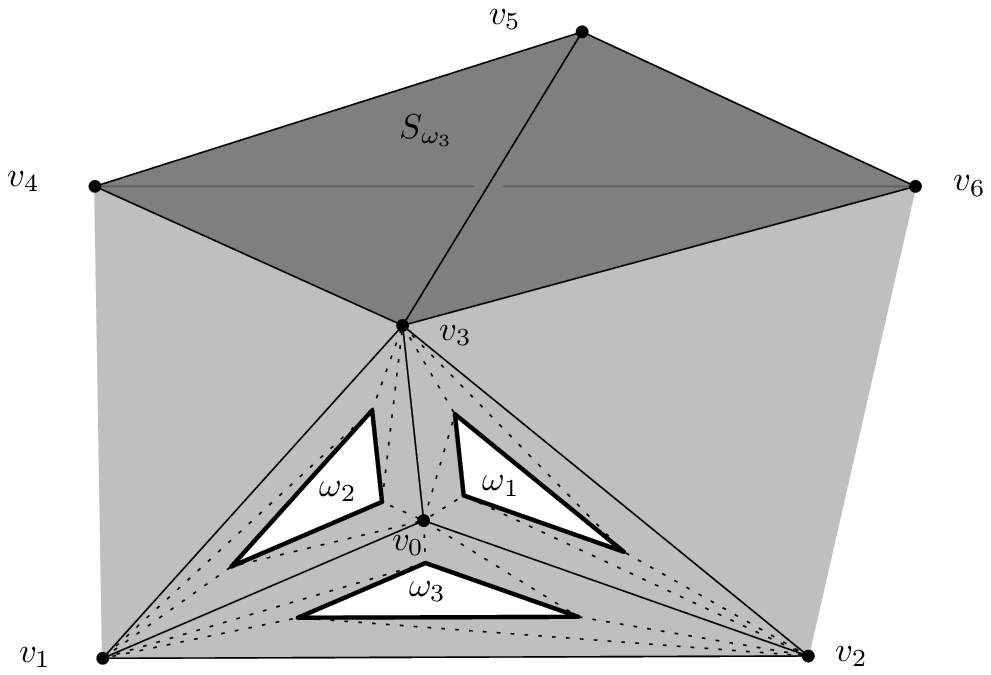}{The clause gadget $\CG$, its openings, and one
of the complementary spheres.}

Let $\omega_{1},\omega_{2},\omega_{3}$ be the three
small triangles we have removed (where
$\omega_1$ comes from the triangle $v_0v_2v_3$ etc.).
We call them the  \emph{openings}
of $\CG$ and we let 
 $O_\CG:=\{\omega_{1},\omega_{2},\omega_{3}\}$
be the set of openings. Thus, $\CG\cup O_\CG$
is a subdivision of the full $2$-skeleton $F$.

If we remove from $F$ the vertices $v_0,v_1,v_2$
and all simplices containing them, we obtain 
the boundary of the $3$-simplex $\{v_3,v_4,v_5,v_6\}$.
Topologically it is an $S^2$, we call it the
\emph{complementary sphere} of the opening
$\omega_{3}$, and we denote it by $S_{\omega_{3}}$.
The complementary spheres of the openings $\omega_{1}$
and $\omega_{2}$ are defined analogously.
The following lemma is a variation on results in Van Kampen
\cite{vanKampen}:

\begin{lemma}\label{l:cgadget}\ 
\begin{enumerate}
\item[\rm(i)] 
For every generic
 PL embedding $f$ of $\CG$ into $\R^4$
there is at least one opening $\omega\in O_\CG$ such that
the images of the boundary  $\bd\omega$
 and of the complementary sphere
$S_\omega$ have odd linking number.
\item[\rm(ii)] For every opening $\omega\in O_\CG$ 
there exists an embedding of $\CG$ into $\R^4$ in which
only $\bd\omega$ is linked with its complementary sphere.
More precisely, there exists a generic linear mapping
of the full $2$-skeleton $F$ into $\R^4$
whose restriction to $\polyh{\CG\cup O_\CG\setminus\{\omega\}}$
is an embedding.  
\end{enumerate}
\end{lemma}

\begin{proof}[Proof of (i)] This is very similar to Lemma~6 in \cite{FKT}.
Let $f_0$ be a generic PL map
(not necessarily an embedding) of $F$ into $\R^4$. 
Van Kampen proved that
$$\sum_{\{\sigma,\tau\}} |f_0(\sigma)\cdot f_0(\tau)|
$$ 
is always odd,
where $|f_0(\sigma)\cdot f_0(\tau)|$ 
denotes the number of intersections
between the image of $\sigma$ and the image of $\tau$, and the sum
is over all unordered pairs of \emph{disjoint} $2$-dimensional
simplices $\sigma,\tau\in F$ (the genericity of $f_0$ guarantees
that the intersection $f_0(\sigma)\cap f_0(\tau)$ consists of finitely
many points). (See Appendix~\ref{s:vkamp} for
a wider context of this result.)

Now let us consider a generic PL embedding $f$ of $\CG$ into $\R^4$,
and let us extend it piecewise linearly and generically (and otherwise
arbitrarily) to the openings of $\CG$.
The resulting map can also be regarded as a generic PL map  
 $f_0$ of $F$ into $\R^4$.
For such an $f_0$, $|f_0(\sigma)\cdot f_0(\tau)|$ can be nonzero
only if $\sigma$ contains an opening $\omega$ of
 $\CG$ and $\tau$ belongs
to its complementary sphere $S_\omega$ (or the same situation
with $\sigma$ and $\tau$ interchanged). Thus, for at least
one $\omega\in O_\CG$, $f_0(\omega)$ intersects $f(S_\omega)$
in an odd number of points, and this means exactly that
$f(\bd\omega)$ and $f(S_\omega)$ have odd linking number.\end{proof}

\begin{proof}[Proof of (ii)] It suffices to exhibit a generic linear 
map $f_0$ of $F$ into $\R^4$ such that the images of
two disjoint $2$-simplices intersect (at a single point),
and this intersection is the only multiple point of $f_0$.
Such a mapping was constructed by Van Kampen \cite{vanKampen}:
$5$ of the vertices are placed as vertices of a $4$-dimensional
simplex in $\R^4$, and the remaining two are mapped
in the interior of that simplex. 
\end{proof}

\subsection{The conflict gadget}\label{ss:conflgadget}
\lipefig{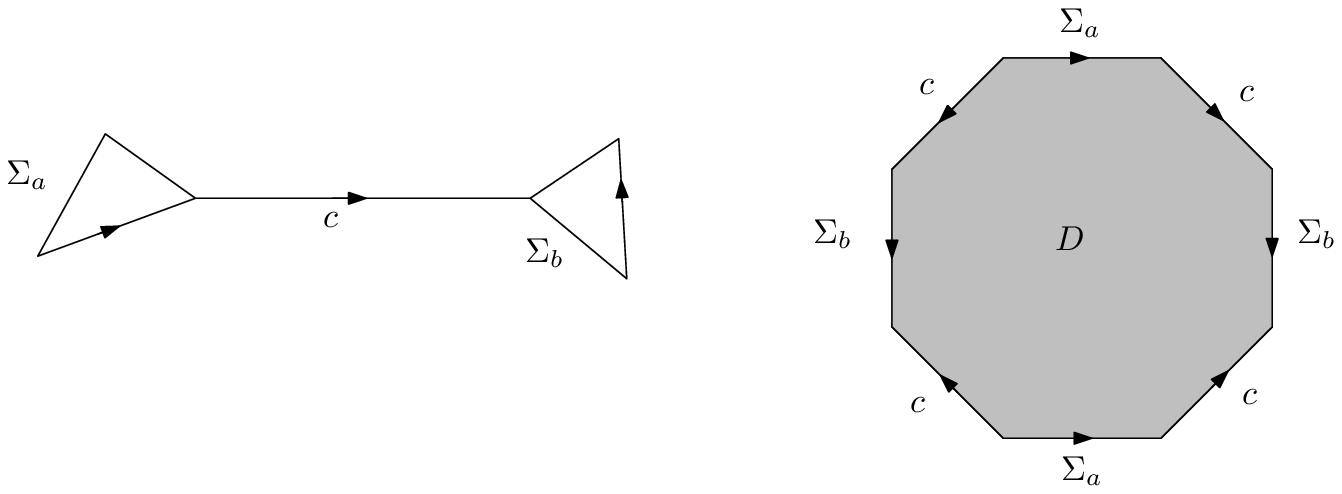}{Attaching a disk to the polygonal line $E$.}
To construct $\TG$, we start with the $1$-dimensional simplicial
complex $E$ shown in Fig.~\ref{f:e4h-attach} left,
consisting of two triangular loops $\Sigma_a$ and $\Sigma_b$ and
an edge $c$ connecting them. We also fix an orientation
of $\Sigma_a$, $\Sigma_b$, and $c$ (marked by arrows). 
Then we take a disk $D$
and we attach its boundary to $E$ as indicated in Fig.~\ref{f:e4h-attach} 
right; the disk is triangulated sufficiently finely so that the result
of the attachment is still a simplicial complex. This is the
complex $\TG$.

We observe that topologically, $\TG$ is a ``squeezed torus'' (the reader may want to recall the
usual construction of a torus by gluing the opposite sides of
a square; this well-known construction would be obtained
from the attachment as above if the edge $c$ were contracted
to a point). Fig.~\ref{f:e4h-tor} shows such a squeezed torus
embedded in $\R^3$ (with the loops $\Sigma_a$ and $\Sigma_b$ drawn circular
rather than triangular).
\lipefig{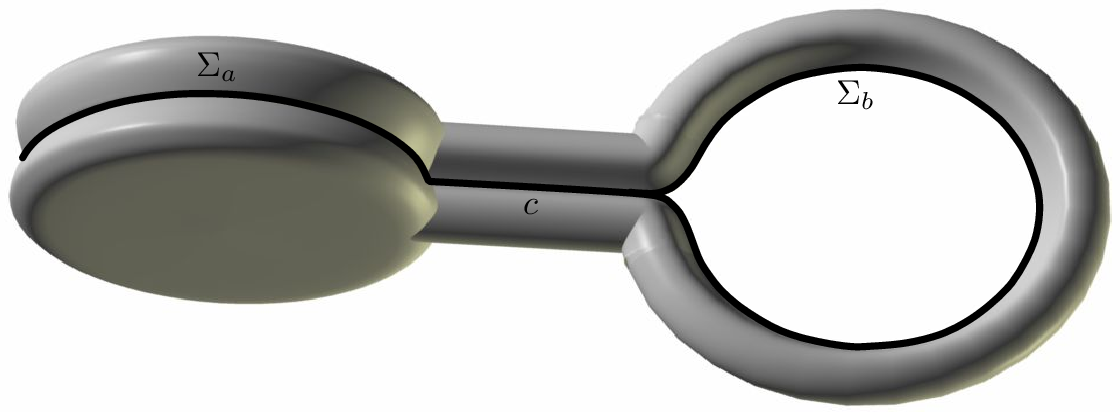}{A $3$-dimensional embedding of the conflict gadget. }

\begin{lemma}\label{l:confl-gadget}\ 
\begin{enumerate}
\item[\rm(i)]{\rm  \cite[Lemma~7]{FKT}}
Let $S_a$ and $S_b$ be PL $2$-spheres.
Then there is no PL embedding $f$ of $S_a\sqcup S_b\sqcup \TG$
(disjoint union) into $\R^4$ such that 
\begin{itemize}
\item
the $1$-sphere $f(\Sigma_a)$ and the $2$-sphere $f(S_a)$ have
odd linking number, 
and so do $f(\Sigma_b)$ and $f(S_b)$;
\item $f(\Sigma_a)$ and $f(S_b)$ are unlinked,
 and so
are $f(\Sigma_b)$ and $f(S_a)$.
\end{itemize}
\item[\rm(ii)] Let $f$ be a generic linear embedding of
$E$ in $\R^3$ (not $\R^4$ this time)  such that
$f(\Sigma_a)$ and $f(\Sigma_b)$ are unlinked, and let $\delta>0$. 
Then there is a PL embedding $\overline f$ of $\TG$ in $\R^3$
extending $f$ whose image is contained in the set
$N=N(f,\delta):=N(T_a,\delta)\cup N(f(\Sigma_b),\delta)\cup N(f(c),\delta)$,
where $T_a$ is the triangle bounded by the loop $f(\Sigma_a)$ and
$N(A,\delta)$ denotes the $\delta$-neighborhood of a set $A$
(in $\R^3$ in our case).\footnote{Formally
$N(A,\delta)=\{x\in\R^3: \dist(x,A)\le\delta\}$,
where $\dist(x,A)$ is the Euclidean distance of $x$
from the set~$A$.}
 (Symmetrically, and this is the main
point of the construction, we can also embed
$\TG$ into $N(f(\Sigma_a),\delta)\cup N(T_b,\delta)\cup N(f(c),\delta)$,
thus leaving a hole on the other side.)
\end{enumerate}
\end{lemma}

For a proof of part (i) we refer to \cite{FKT}
 (a few words about the basic approach
of the proof will be said in the proof of
Lemma~\ref{l:confl-gadget-higher} below), and for part
(ii) to Fig.~\ref{f:e4h-tor}.

\subsection{The reduction}\label{ss:reduct}

 Let the given $3$-CNF formula be $\varphi=C_1\wedge C_2\wedge\cdots\wedge C_m$,
where each $C_i$ is a clause with three literals (each literal
is either a variable or its negation). For each $C_i$, we
take a copy of the clause gadget $\CG$ and we denote
it by $\CG_i$ (the $\CG_i$ have pairwise disjoint vertex
sets). We fix a one-to-one correspondence between the literals
of $C_i$ and the openings of $\CG_i$, letting $\omega(\lambda)$
be the opening corresponding to a literal $\lambda$.

Let us say that a literal $\lambda$ in a clause $C_i$
\emph{is in conflict} with a literal $\mu$ in a clause
$C_{j}$ if both $\lambda$ and $\mu$ involve the same variable $x$
but one of them is $x$ and the other the negation $\overline x$.
For convenience
we assume, without loss of generality,
that two literals from the same clause are never in conflict. 

Let $\Xi$ consist of all (unordered) pairs
$\{\omega(\lambda),\omega(\mu)\}$ of openings corresponding to
pairs $\{\lambda,\mu\}$ of conflicting literals in $\varphi$.
For every pair $\{\omega,\psi\}\in\Xi$ 
we take a fresh copy $\TG_{\omega\psi}$ of the conflict
gadget $\TG$. We identify the loop $\Sigma_a$ in $\TG_{\omega\psi}$
with the boundary $\bd\omega$  and the loop $\Sigma_b$ 
with $\bd\psi$ (the rest of $\TG_{\omega\psi}$
is disjoint from the clause gadgets and the other conflict gadgets).

The simplicial complex $K$ assigned to the formula $\varphi$
is
$$ K:=
\biggl(\,\bigcup_{i=1}^m \CG_i\biggr)\cup
\biggl(\,\bigcup_{\{\omega,\psi\}\in\Xi} \TG_{\omega\psi}\biggr).
$$
It remains to show that $K$ is PL embeddable in $\R^4$ exactly if $\varphi$
is satisfiable.

\heading{Nonembeddability for unsatisfiable formulas. }
This is a straightforward consequence of Lemma~\ref{l:cgadget}(i)
and Lemma~\ref{l:confl-gadget}(i). 

Indeed, if $f$ is a PL embedding
of $K$ into $\R^4$, which we may assume to be generic,
there is an opening  in each clause gadget $\CG_i$
such that $f(\bd\omega_i)$ has odd linking number with
the complementary sphere $f(S_{\omega_i})$; let us call
it a \emph{occupied opening} of $\CG_i$. Since $\varphi$ is not
satisfiable, whenever we choose one literal from each clause,
there are two of the chosen literals in conflict. Thus,
there are two occupied openings $\omega\in O_{\CG_i}$ and
$\psi\in O_{\CG_j}$ that are connected by a conflict gadget
$\TG_{\omega\psi}$.

Then the supposed PL embedding $f$ provides us an embedding
as in Lemma~\ref{l:confl-gadget}(i) with
$S_a=S_\omega$, $S_b=S_\psi$, and $\TG=\TG_{\omega\psi}$.
Concerning the assumptions in the lemma, we already
know that $f(S_\omega)$ and $f(\bd\omega)$ have odd linking
number, and so do $f(S_\psi)$ and $f(\bd\psi)$. It remains
to observe that $f(\bd\omega)$ cannot be linked with
$f(S_\psi)$ (and vice versa), since $\CG_i$ contains a 
disk bounded by $\bd\omega$: For example (refer to Fig.~\ref{f:e4h-C}),
$\bd\omega_{3}$ is the boundary of the disk consisting
of the triangles $v_0v_1v_4$, $v_0v_2v_4$, $v_1v_2v_4$
and the triangles in the subdivision of $v_0v_1v_2$ different
from $\omega_{3}$. So the lemma applies and $K$ is not embeddable.

\heading{Embedding for satisfiable formulas. } Given a satisfying
assignment for $\varphi$, we choose a \emph{witness literal}
$\lambda_i$ for each clause $C_i$ that is true under the given
assignment (and we will refer to the remaining
two literals of $C_i$ as \emph{non-witness} ones). 
No two witness literals can be in conflict. 

We describe an embedding of $K$ into $\R^4$ corresponding to
this choice of witness literals.

Let us choose  distinct points $p_1,\ldots,p_m\in\R^4$.
For each $i=1,2,\ldots,m$, we let $f_i$ be a generic linear embedding of
the clause gadget $\CG_i$ into a small neighborhood of $p_i$
(and far from the other $p_j$)
as in Lemma~\ref{l:cgadget}(ii), where the role
of $\omega$ in the lemma is played by the
witness opening of $\CG_i$ (i.e., the one corresponding to
to the witness literal of $C_i$). 
In particular, the interiors of the triangles bounded by
$f_i(\bd\omega')$ and by $f_i(\bd\omega'')$ are disjoint
from $f_i(\CG_i)$, where $\omega'$ and $\omega''$ are the
non-witness openings of $\CG_i$.

Taking all the $f_i$ together defines an embedding $f$
of the union of the clause gadgets,
and it remains to embed the conflict gadgets.

To this end, we will assign to each conflict gadget
$\TG_{\omega\psi}$ a ``private'' set $P_{\omega\psi}\subset\R^4$
homeomorphic to the $3$-dimensional set $N$ from Lemma~\ref{l:confl-gadget}(ii),
and we will embed $\TG_{\omega\psi}$ into $P_{\omega\psi}$.
Each $P_{\omega\psi}$ will be disjoint from 
all other $P_{\omega'\psi'}$ and also from all the images
$f(\CG_i)$, \emph{except} that $P_{\omega\psi}$ has to contain
the loops $f(\bd\omega)$ and $f(\bd\psi)$ where the conflict
gadget $\TG_{\omega\psi}$ should be attached.
In order to fit enough almost-disjoint homeomorphic copies of $N$ 
into the space, we will ``fold'' them suitably.  

We know that for every pair $\{\omega,\psi\}$ of openings
connected by a conflict gadget, at least one of $\omega$ and $\psi$
is non-witness. 
Let us choose the notation so that
$\omega$ is  non-witness and thus unoccupied in the embedding $f$. 

We will build $P_{\omega\psi}$ from three pieces:
a set $Q^+_{\omega\psi}$ that plays the role
of $N(T_a,\delta)$ in Lemma~\ref{l:confl-gadget}(ii),
a set $Q_{\psi\omega}$ that plays the role
of $N(f(\Sigma_b),\delta)$, and  a ``connecting ribbon''
in the role of $N(f(c),\delta)$. 

Now let $\omega$ be an opening of some $\CG_i$, witness or non-witness.
Let  $t$ be the number of openings $\psi$ that
are connected to $\omega$ by a conflict gadget.
The sets $Q_{\omega\psi}$ and $Q^+_{\omega\psi}$
we want to construct are indexed by these $\psi$,
but with some abuse of notation, we will now
regard them as  
indexed by an index $j$ running from $1$ to $t$,
i.e., as $Q_{\omega1}$ through $Q_{\omega t}$
(and similarly for $Q^+_{\omega\psi}$).

For concise notation let us write $\Sigma=f(\bd\omega)$ and
let $T$ be the triangle in $\R^4$ having $\Sigma$ as the boundary.
Let $\eps>0$ be a parameter and let $T^\eps:=\{x\in T:
\dist(x,\bd T)\le\eps \}$ be the part of $T$
at most $\eps$ away from the boundary of $T$.
Since the subdivided triangle in $\CG_i$ containing $\omega$
in its interior is embedded linearly by $f$, there is an
$\eps>0$ such that if we start at a point $x\in T^\eps$
and go distance at most $\eps$ in a direction orthogonal
to $T$, we do not hit $f(\CG_i)$. Moreover,
if $\omega$ is non-witness and thus all of $T$ is
free of $f(\CG_i)$, we can take any $x\in T$ with
the same result. Fig.~\ref{f:e4h-nbh} tries to illustrate
this in dimension one lower, where we have a segment $T$
in $\R^3$ instead of a triangle $T$ in $\R^4$.
\lipefig{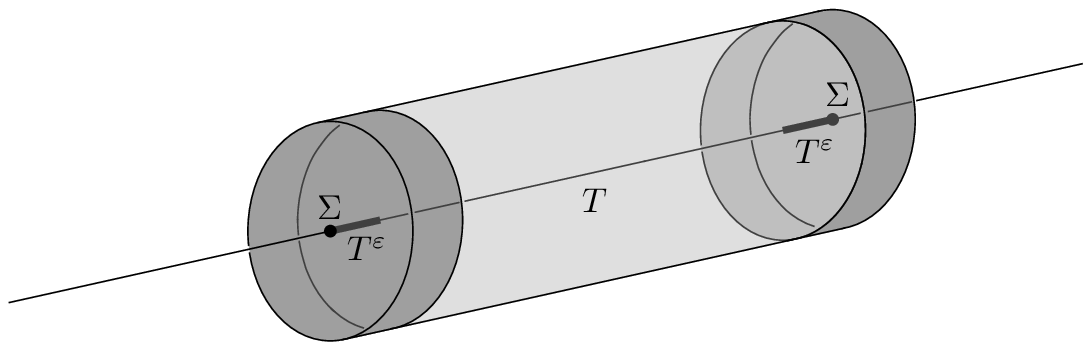}{A free region around the triangle $T$; illustration
in $\R^3$ instead of $\R^4$.}
Thus, there are a set
$Q_\omega\subset\R^4$ with $Q_\omega\cap f(\CG_i)=\Sigma$
and a homeomorphism (actually, a
linear isomorphism) $h\: Q_{\omega}\to T^\eps \times B^2$
with $h(T^\eps)=T^\eps\times \{0\}$, where $0$
is the center of the disk $B^2$. Similarly, if $\omega$
is non-witness, there are $Q^+_\omega$ and $h^+\: Q^+_\omega\to
T\times B^2$ with $h^+(T)=T\times\{0\}$.

\lipefig{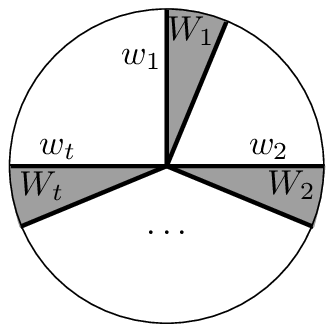}{The wedges.}
Let $W_1,\ldots,W_t\subset B^2$ be disjoint wedges 
as in Fig.~\ref{f:e4h-wedges}, and let $w_j$ consist
of the two radii bounding $W_j$. We set
$$
Q_{\omega j}:= h^{-1}((\Sigma\times W_j)\cup (T^\eps\times w_j)),\ \ \
Q^+_{\omega j} := (h^+)^{-1}((\Sigma \times W_j)\cup (T\times w_j)).
$$
As Fig.~\ref{f:e4h-prod} tries to illustrate,
$Q^+_{\omega j}$  is homeomorphic to a 3-dimensional neighborhood
of $T$ (by a homeomorphism sending $T$ to $T$),
and $Q_{\omega j}$ is similarly homeomorphic to a $3$-dimensional
neighborhood of $\Sigma$. Thus, the sets $Q_{\omega j}$ and
$Q^+_{\omega j}$ can indeed play the roles of
$N(f(\Sigma_b),\delta)$ and $N(T_a,\delta)$, respectively, in Lemma~\ref{l:cgadget}(ii).
\lipefig{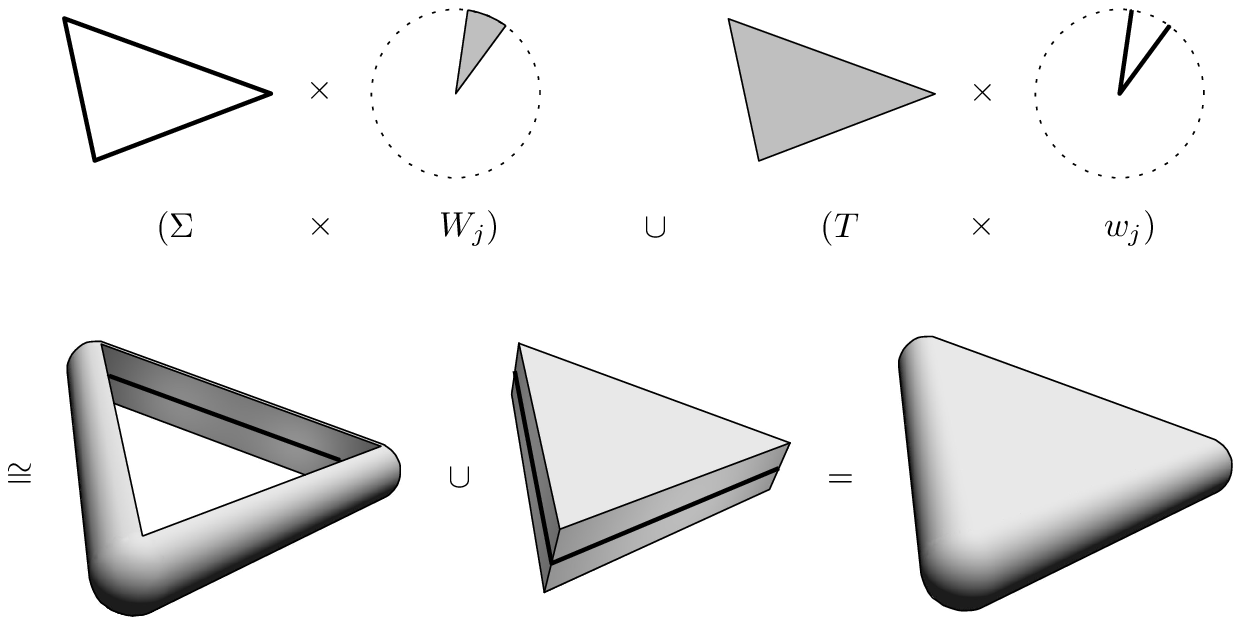}{Folding a $3$-dimensional neighborhood in $\R^4$.}

It remains to construct the ``connecting ribbons'':
For every conflict gadget $\TG_{\omega\psi}$,
we want to connect a vertex of $f(\bd\omega)$ to
a vertex of $f(\bd\psi)$ by a narrow $3$-dimensional ``ribbon''
(it need not be straight since we are looking only for
PL homeomorphic copies of $N$).

We observe that each of the sets $Q_{\omega j}$ and
$Q^+_{\omega j}$ can be deformation-retracted to
the corresponding loop $f(\bd\omega)$ or to
the corresponding triangle, respectively.
It follows that the complement of the union $U$
of all the $Q_{\omega j}$, $Q^+_{\omega j}$,
and $f(\CG_i)$ is path-connected (formally, this
follows from Alexander duality, since this union
is homotopy equivalent to a $2$-dimensional space).
Since all the considered embeddings are piecewise linear,
any two points on the boundary of $U$ can be connected
by a PL path within $\R^4\setminus U$.

Thus, the $3$-dimensional ``ribbon'' connecting
$f(\bd\omega)$ to $f(\bd\psi)$ can first go within
the appropriate $Q_{\omega j}$ to a point on the boundary,
then continue along a path connecting this boundary point
to a boundary point of $Q_{\psi j'}$, and then
reach $f(\bd\psi)$ within $Q_{\psi j'}$. 

In this way, we have allocated the desired ``private'' sets
$P_{\omega\psi}$ for all conflict gadgets $\TG_{\omega\psi}$,
and hence $K$ can be
PL embedded in $\R^4$ as claimed.
This finishes the proof of the special case
$k=2,d=4$ of Theorem~\ref{t:4hard}.\proofend

\section{\NP-hardness for higher dimensions}\label{s:nphard}

In this section we prove all the remaining cases of Theorem~\ref{t:4hard}.
The proof is generally very similar to the case $k=2$, $d=4$ treated
above: We will again reduce 3-SAT using clause gadgets and conflict
gadgets, but the construction of the gadgets and of their embeddings
require additional work.

By the monotonicity of \EMBED kd in $k$ mentioned in Section~\ref{s:intro},
it suffices to consider $d\ge 5$ and $k=\lceil (2d-2)/3\rceil$.
In the construction we will often use the integer
$\ell:= d-k-1$.

\subsection{The clause gadget}

The clause gadget $\CG=\CG(k,\ell)$
is very similar to a construction of Segal and
Spie\.z~\cite{segal-spiez}. We use the parameters $k,\ell,d$
as above. For the purposes of the present section we
need that $1\le \ell<k$ and $d-\ell=k+1\ge 3$ (which are
easy to verify using the definitions of $k$ and $\ell$ and
the assumption $d\ge 5$).

For the parameters $k,\ell,d$ as above,
we first define a simplicial complex $F=F(k,\ell)$ on the
vertex set $V:=\{v_0,v_1,\ldots,v_{d+1},p\}$ as the union
$F:=F_0\cup C_p$ of the following two sets of simplices:
\begin{itemize}
\item $F_0$ is the $k$-skeleton of the $(d+1)$-simplex
with vertex set $\{v_0,\ldots,v_{d+1}\}$;
\item $C_p$ consists of all the $(\ell+1)$-dimensional simplices on $V$ that contain $p$.
\end{itemize}
See Fig.~\ref{f:e4h-Fkl} for a schematic illustration; let us also
note that for $d=4$, $k=2$, $\ell=1$ we would get exactly the $F$
as in Section~\ref{ss:clgadget}.
\lipefig{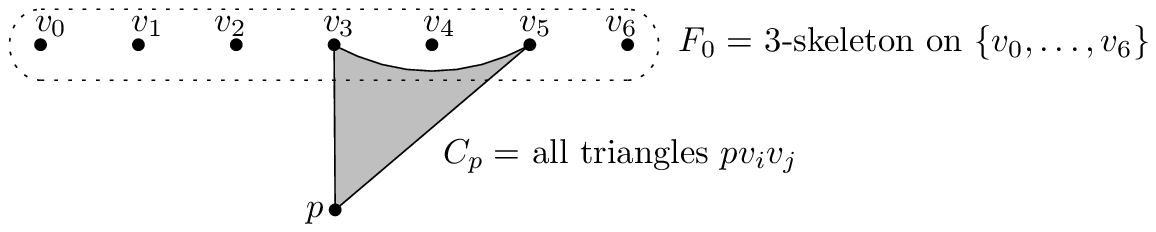}{A schematic illustration of $F(3,1)$.}

Let us consider some $\sigma\in C_p$. 
By removing from $F$ all simplices intersecting 
$\sigma$ (including $\sigma$),
we obtain  the $k$-skeleton
of a $(k+1)$-simplex, i.e., an $S^k$, which we call
the \emph{complementary sphere} $S_\sigma$.

Next, we fix three $(\ell+1)$-dimensional simplices
$\sigma_1,\sigma_2,\sigma_3\in C_p$, say
$\sigma_1:=pv_0v_2v_3\cdots v_{\ell+1}$, $\sigma_2:=
pv_0v_1v_3\cdots v_{\ell+1}$, and $\sigma_3:= pv_0v_1v_2v_4\cdots v_{\ell+1}$.
As in Section~\ref{ss:clgadget}, we make a
hole in the interior of each $\sigma_i$, i.e., we subdivide each
$\sigma_i$, $i=1,2,3$, and we remove a small $(\ell+1)$-simplex 
$\omega_i$ in the middle.
This yields the simplicial complex $\CG = \CG(k,\ell)$.

The $\omega_i$ are again called the  \emph{openings} of $\CG$,
and we set $O_G := \{ \omega_1,\omega_2, \omega_3\}$.
The \emph{complementary sphere} $S_{\omega_i}$ is
defined, with some abuse of notation, as the complementary
sphere of the simplex $\sigma_i\in C_p$ that contains $\omega_i$. 

\begin{lemma}[Higher-dimensional version of Lemma~\ref{l:cgadget}]\
\label{l:clause-gadget-higher}
\begin{enumerate}
\item[\rm(i)]
For every generic
PL embedding $f$ of $\CG$ into $\R^d$
there is at least one opening $\omega\in O_\CG$ such that
the images of the boundary  $\bd\omega$
 and of the complementary sphere
$S_\omega$ have odd linking number.
\item[\rm(ii)] For every opening $\omega\in O_\CG$
there exists a generic linear embedding of $\CG$ into $\R^d$ in which
the boundaries of the two openings different from $\omega$
are unlinked with their complementary spheres.
\end{enumerate}
\end{lemma}

\begin{proof}
Part (ii) is established in the proof of Lemma~1.1 in Segal and
Spie\.z~\cite{segal-spiez} (generalizing Van Kampen's
embedding mapping mentioned in the proof
of Lemma~\ref{l:cgadget}(ii)). They construct a PL
embedding of $F(k,\ell)$ (which they call $P(k,\ell)$,
while their $n$ is our $d-1$), but inspecting the 
first two paragraphs of their proof
reveals that their embedding is actually linear
(in the subsequent paragraphs,
 they modify the embedding on the interior
of one of the $(\ell+1)$-simplices from $C_p$,
but this serves only to show the claim about linking
number).
\medskip

For part (i), it clearly suffices to prove the following:
\begin{quotation}
\emph{{\bf Claim. }
For any generic PL mapping $g$ of $F$ in $\R^d$
whose restriction to $F_0$ is an embedding,
there is an $(\ell+1)$-dimensional simplex $\sigma\in C_p$ 
 such that $|g(\sigma)\cap g(S_\sigma)|$ is odd.
 } 
\end{quotation}

This claim follows easily from the \emph{proof} of
Lemma~1.4 in Segal and
Spie\.z \cite{segal-spiez}. Indeed, they give a procedure
that, given a generic PL map $g_1$ of $F$ into $\R^d$
such that $|g_1(\sigma)\cap g_1(S_\sigma)|$ is even for some
$\sigma$,
constructs a new generic PL map $g_2$ with
$g_2(\sigma)\cap g_2(S_\sigma)=\emptyset$ and such that
there are no new intersections between images of
\emph{disjoint} simplices (compared to $g_1$).\footnote{The
procedure requires $d-\ell\ge 3$, which is satisfied in our case.
In \cite{segal-spiez} this inequality is reversed by mistake.}

Assuming that there is a $g$ contradicting the claim, after finitely
many applications of the procedure we arrive at a generic
PL mapping $\tilde g$ such that $\tilde g(\sigma)\cap 
\tilde g(S_\sigma)=
\emptyset$ for every $\sigma\in C_p$. We claim
that then 
\begin{equation}\label{e:tautau}
\tilde g(\tau)\cap \tilde g(\tau')=\emptyset\mbox{ for every 
$\tau,\tau'\in F$ with $\tau\cap\tau'=\emptyset$}.
\end{equation}
Indeed, if (\ref{e:tautau}) fails for some
$\tau,\tau'$, one of $\tau,\tau'$ (say $\tau'$) must belong to
$C_p$, since $g$ restricted to $F_0$ is an embedding.
But then we get $\tau\in S_{\tau'}$---a contradiction.
Hence (\ref{e:tautau}) holds.
But no generic PL mapping $\tilde g$ satisfying (\ref{e:tautau})
exists according to \cite{segal-spiez} (end of the proof
of Lemma~1.4). This proves the claim and thus
also part (i) of Lemma~\ref{l:clause-gadget-higher}.

Let us remark that a perhaps more conceptual proof of part (i)
can be obtained using the results of Shapiro~\cite{Shapiro-realiz2d}
on the ``generalized Van Kampen obstruction'', but we would need
many preliminaries for presenting it. 
\end{proof}

\subsection{The conflict gadget}

Here we construct the conflict gadget $\TG=\TG(\ell)$, 
which depends only on the parameter $\ell$, and whose
dimension is $2\ell$. The conflict gadget
$\TG$ in Section~\ref{ss:conflgadget} is essentially
the same as the following construction for 
$\ell=1$, up to minor formal differences.
In addition to
the inequalities among the parameters
 mentioned earlier, here we 
also need $2\ell\le k$ (which again holds
in our setting).

In the $\ell=1$ case we attached a $2$-dimensional disk by its
boundary to the $1$-dimensional complex $E$. The $\ell$-dimensional
version of $E$ consists of two disjoint copies $\Sigma^\ell_a$ and $\Sigma^\ell_b$ of the boundary
of the $(\ell+1)$-simplex connected by an edge $c$ (see Fig.~\ref{f:e4h-E1}).
\lipefig{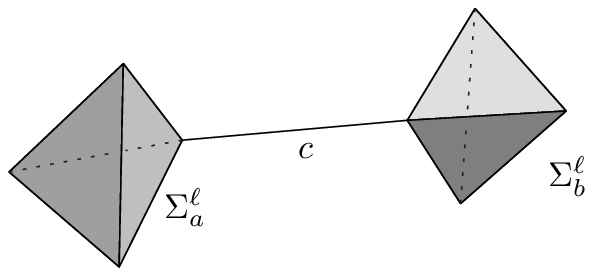}{A higher-dimensional version of $E$.}
To this $E$ we are going to
attach the $(2\ell)$-dimensional ball $B^{2\ell}$  by its
boundary. For $\ell=1$ the result was topologically a ``squeezed'' version
of the $2$-dimensional torus $S^1\times S^1$;
for larger $\ell$ it is going to be the higher-dimensional
``torus'' $S^\ell\times S^\ell$, again suitably squeezed.

\heading{Attaching a ball to \boldmath$S^\ell\vee S^\ell$. }
Before defining $\TG$ itself, we 
define a certain mapping $g\:S^{2\ell-1}\to
S^\ell\vee S^\ell$, where $S^\ell\vee S^\ell$ is a wedge
of two spheres, to be defined below. This construction
is based on  the \emph{Whitehead product} 
in homotopy theory.
As we will see, attaching the boundary of $B^{2\ell}$
to $S^\ell\vee S^\ell$ via $g$ results topologically
in $S^\ell\times S^\ell$ (without any squeezing).

The wedge $S^\ell\vee S^\ell$ consists of
two copies of the sphere $S^\ell$ glued together at one point.
For our purposes, we represent $S^\ell\vee S^\ell$ concretely as follows.
We consider $S^\ell$ geometrically as the unit sphere in $\R^{\ell+1}$,
we choose a distinguished point $s_0=(1,0,0,\ldots,0)\in S^{\ell}$,
and we let $S^\ell\vee S^\ell$ be the subspace
$(S^\ell\times \{s_0\})\cup(\{s_0\}\times S^\ell)$ of
$\R^{\ell+1}\times \R^{\ell+1}=\R^{2\ell+2}$.
For $\ell=1$, we thus get two unit circles lying in perpendicular
$2$-flats in $\R^4$ and meeting at the point $(s_0,s_0)$.

For defining the map $g$, we need to represent the
ball $B^{2\ell}$ not as the standard Euclidean unit ball,
but rather as the product $B^\ell\times B^\ell$
(which is clearly homeomorphic to $B^{2\ell}$).
Then we have
\begin{equation}\label{e:s2l-1repr}
S^{2\ell-1}\cong \bd(B^\ell\times B^\ell)=
(B^\ell\times S^{\ell-1})\cup(S^{\ell-1}\times B^{\ell});
\end{equation}
see the left part of 
Fig.~\ref{f:e4h-sq} for the (rather trivial) case $\ell=1$.
(Indeed, for arbitrary sets $A\subseteq \R^m$ and $B\subseteq \R^n$
we have $\bd(A\times B)=(A\times \bd B)\cup (\bd A\times B)$,
as is easy to check.)

\lipefig{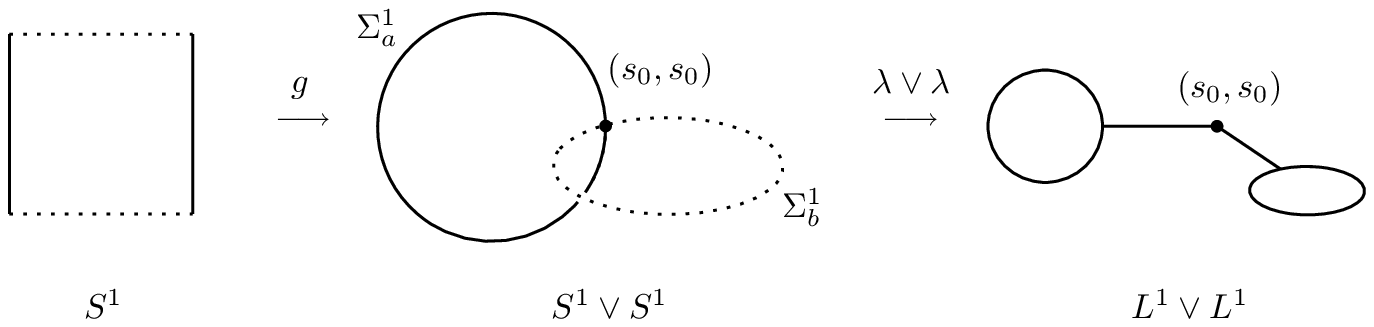}{Representing
$S^1$ as $(B^1\times S^0)\cup(S^0\times B^1)$ (left);
mapping it to $S^1\vee S^1$ (middle); squeezing
the $S^1$'s to ``lollipops'' (right).
We note that $S^1\vee S^1$ and $L^1\vee L^1$ actually
live in $\R^4$.}

As is well known, if we shrink the boundary of an $n$-ball to
a single point, the result is an $n$-sphere. Let us fix a mapping
$\gamma\: B^{\ell}\to S^{\ell}$ that sends all of $\bd B^{\ell}$
to the distinguished point $s_0$ and is a homeomorphism on
the interior of $B^\ell$. Now we are ready to define the map $g$.
Namely, we define
$\overline g\:B^{2\ell}\to S^\ell\times S^\ell$ by
$$
\overline g(x,y)=(\gamma(x),\gamma(y)),
$$
where we still consider $B^{2\ell}$ as $B^\ell\times B^\ell$
and $x$ comes from the first $B^\ell$ and $y$ from the second.
Then $g$ is the restriction of $\overline g$ to $S^{2\ell-1}=\bd B^{2\ell}$.

For the image of $g$ we have, using (\ref{e:s2l-1repr}),
$$
g(S^{2\ell-1})=g(B^\ell\times S^{\ell-1})\cup
g(S^{\ell-1}\times B^{\ell})=
(S^\ell\times\{s_0\})\cup(\{s_0\}\times S^{\ell})=S^\ell\vee S^\ell.
$$
It remains to observe that $\overline g$ restricted to $\interior B^{2\ell}$
is a homeomorphism onto $(S^\ell\times S^{\ell})\setminus
(S^\ell\vee S^\ell)$. Hence, the result of attaching
the boundary of $B^{2\ell}$ to $S^\ell\vee S^\ell$ via $g$
is indeed homeomorphic to $S^\ell\times S^{\ell}$ as claimed.

\heading{Squeezing. } Now we define a ``squeezing  map''
from $S^\ell\vee S^\ell$ to $E$. We let
the \emph{$\ell$-lollipop} $L^\ell$ be an $\ell$-dimensional
sphere of radius $\frac12$ with attached segment
(``stick'') of length $1$; see Fig.~\ref{f:e4h-loll}.
\lipefig{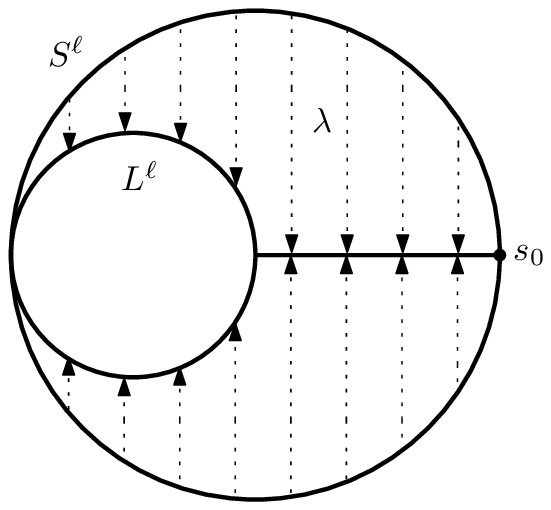}{The map $\lambda$ squeezing $S^\ell$ to the lollipop $L^\ell$.}
Formally, $L^\ell:=\bd B(-\frac12s_0,\frac12)\cup[0,s_0]$,
where $B(x,r)$ stands for the ball of radius $r$ centered
at $x$.
We let $\lambda\:S^\ell\to L^\ell$ be the projection that
moves each point of $S^\ell$ in direction perpendicular
to the axis $[-s_0,s_0]$.

Now, with $L^\ell\vee L^\ell:= (L^\ell\times\{s_0\})\cup
(\{s_0\}\times L^\ell)$, we have the map
$\lambda\vee\lambda\:S^\ell\vee S^\ell\to L^\ell\vee L^\ell$
(given by $(x,y)\mapsto (\lambda(x),\lambda(y))$).
Finally, $L^\ell\vee L^\ell$
can be identified with the complex $E$ as above by 
a suitable homeomorphism, and we arrive at the map
$$
r=(\lambda\vee\lambda)\circ g\: S^{2\ell-1}\to E
$$
(where the homeomorphism of
$L^\ell\vee L^\ell$ with $E$ is not explicitly shown). 

The clause gadget $\TG$ is obtained by attaching the boundary
of $B^{2\ell}$ to $E$ via the map~$r$. Of course, we want $\TG$
to be a simplicial complex, and so in reality we use a
suitable PL version of the attaching map $r$ (we have not presented
it this way since the description above seems more accessible).

For the forthcoming  proof of an analogue of
Lemma~\ref{l:confl-gadget}, we need the following observation. 

\begin{obs} 
\label{obs:collapse-hemispheres}
Let $\kappa\: L^\ell \vee L^\ell \rightarrow S^\ell\vee S^\ell$ be the 
quotient map corresponding to contracting the ``stick'' $c$ of the 
double-lollipop to a single point. Then the composition 
$\kappa \circ (\lambda \vee \lambda)\: S^\ell \vee S^\ell \to 
S^\ell\vee S^\ell$ is homotopic to the identity 
on $S^\ell \vee S^\ell$. 
\end{obs}
\begin{proof}
Let $H^\ell:=\{ x \in S^\ell :
 \langle s_0,x\rangle\geq 0\}$ 
be the closed hemisphere centered at $s_0$.
The assertion follows by observing that $\kappa \circ (\lambda \vee \lambda)$ is the quotient map corresponding to contracting the subset $H^\ell \vee H^\ell$ of $S^\ell\vee S^\ell$ to a single point; see Fig.~\ref{f:e4h-hemispherecollapse}.
\lipefig{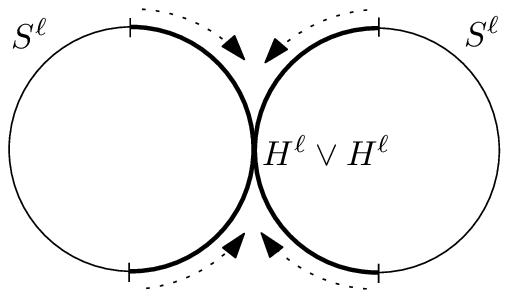}{Contracting the wedge of two hemispheres.} 
\end{proof}

\begin{lemma}[Higher dimensional version of Lemma~\ref{l:confl-gadget}]\
\label{l:confl-gadget-higher}
\begin{enumerate}
\item[\rm (i)]{\rm (Based on \cite[Lemma~2.2]{segal-skopenkov-spiez}).} 
Let $\Sigma^\ell_a$ and $\Sigma^\ell_b$ denote the two
$\ell$-spheres (boundaries of $(\ell{+}1)$-simplices) 
contained in $E\subset \TG$. Let $S_a^k$ and $S_b^k$ be PL $k$-spheres. 
Then there is no PL embedding%
$f$ of the disjoint union $S_a^k \sqcup S_b^k \sqcup \TG$ into $S^d$
such that
\begin{itemize}
\item
the $\ell$-sphere $f(\Sigma_a^\ell)$ and the $k$-sphere $f(S_a^k)$ have odd
 linking
number, and so do $f(\Sigma_b^\ell)$ and 
$f(S_b^k)$;
\item
$f(\Sigma^\ell_a)$ and $f(S_b^k)$ are unlinked,
and so are $f(\Sigma^\ell_b)$ and $f(S_a^k)$.
\end{itemize}
\item[\rm (ii)]
Let $f$ be a generic linear embedding of $E$ in%
\footnote{It follows from our assumptions on $d$ and $k$ that
$d \geq 2\ell+3$. Therefore, when, in the course of the reduction,
we construct an embedding of a complex associated with a satisfiable
formula, we can afford to embed each conflict gadget in its own ``private'' 
$(2\ell+2)$-dimensional set. Since two $\ell$-spheres in dimension $2\ell+2$
are never linked, we do not need to make an explicit unlinking assumption
as in Lemma~\ref{l:confl-gadget}.
} $\R^{2\ell + 2}$, and let
$\delta > 0$. Then there is a PL embedding $\overline{f}$ of $X$ in $\R^{2\ell +
2}$ extending $f$ whose image is
contained in the neighborhood $N = N(f, \delta) := N(T_a, \delta) \cup N(f(\Sigma^\ell_b),
\delta) \cup N(f(c), \delta)$, where $T_a$ is the 
$(\ell+1)$-dimensional simplex
bounded by $f(\Sigma_a^\ell)$.
\end{enumerate}
\end{lemma}

\heading{Proof of (i). }
Part (i) follows from 
the proof of \cite[Proof of Lemma~2.2]{segal-skopenkov-spiez} 
with only minor modifications. 
First, before giving a formal proof,
we describe the basic approach
of \cite{segal-skopenkov-spiez}, which also applies to the proof
of Lemma~\ref{l:confl-gadget}(i). 

Suppose that a PL embedding $f$ as in (i) above exists. Let $C$ denote the complement $\R^d\setminus f(S_a^k \sqcup S_b^k)$, let 
$r\:S^{2\ell-1} \to E$
be the attaching map used in the construction of $X$, 
and let $\overline r\:B^{2\ell}\to \TG$ be the extension
of $r$ to $B^{2\ell}$ (formally, $\overline r$ is the quotient map). 

The basic strategy is as follows: 
On the one hand, using the assumptions about linking numbers, 
one shows that $f\circ r$ defines a nontrivial element of the homotopy 
group $\pi_{2\ell-1}(C)$. On the other hand, $f\circ \overline{r}$ 
witnesses that $f\circ r$ is homotopically trivial---a contradiction.

As in \cite{segal-skopenkov-spiez}, one distinguishes two cases: $\ell=1$ and $\ell>1$. 
In the case $\ell=1$, we are dealing with the fundamental group $\pi_1(C)$,
and the proof is essentially identical to that of \cite[Lemma~7]{FKT},
i.e., our Lemma~\ref{l:confl-gadget}, which we briefly summarize
for the reader's convenience.

For showing that $f\circ r\:S^1\to C$ is homotopically nontrivial,
one first observes that $\pi_1(E)$ is the free group on two 
generators $a$ and $b$, and the attaching map $r\:S^1\to E$
corresponds to the commutator $aba^{-1}b^{-1}$, which
is a nontrivial element of $\pi_1(E)$. So it suffices to
show that the map $f_\star\:\pi_1(E)\to \pi_1(C)$ induced by
the restriction $f|_E$  is injective. To this end,
one first considers the homomorphisms $f_{*1}$ and
$f_{*2}$ induced by $f|_E$
in the first and second \emph{homology}.

By Alexander duality, the complement $C$ has the same homology 
(with $\Z_2$-coefficients, say) as $S^1\vee S^1$,  and thus
$H_1(C;\Z_2)=\Z_2 \oplus \Z_2$ and $H_2(C;\Z_2)=0$.
For $E$ we have $H_1(E;\Z_2)\cong \Z_2 \oplus \Z_2$ with a basis represented 
by the two circles $\Sigma^1_a$ and $\Sigma^1_b$. 
The assumption on the linking numbers imply that
$f_{*1}$ is an isomorphism, and $f_{*2}$ is trivially surjective.
Then the injectivity of the homomorphism $f_\star$ of the
fundamental groups follows from
 a theorem of Stallings~\cite{Stallings:HomologyCentralSeriesGroups-1965},
which finishes the case $\ell=1$.

In the case $\ell>1$, the proof that $f\circ r$ defines 
a nontrivial element of $\pi_{2\ell-1}(C)$ requires somewhat more 
advanced machinery. Segal et al.~\cite{segal-skopenkov-spiez} prove 
essentially the same assertion as in part~(i) of the lemma, 
with the following differences: 
\begin{enumerate}
\item $\TG$ is replaced by $X'$, which is obtained 
by attaching $B^{2\ell}$ to $\Sigma^\ell_a \vee \Sigma^\ell_b$ via the map $g$
(as described above) and hence homeomorphic to $S^\ell \times S^\ell$. 
\item The disjoint union $S^k_a\sqcup S^k_b$ is replaced by the wedge%
\footnote{Wedges are used for technical reasons: 
By a theorem of Lickorish~\cite{Lickorish:UnknottingCones-1965}, 
any embedding (PL or even topological)
of a wedge of spheres of codimension at least 3 
is unknotted, i.e., ambient isotopic to a standard embedding.}
$S^k_a \vee S^k_b$. 
\end{enumerate}
They show that if there were an embedding $f$ of 
$(\Sigma^\ell_a\vee \Sigma^\ell_b) \sqcup (S^k_a \vee S^k_b)$ 
with the linking properties as in part (i) of the lemma, 
 $f\circ g$ would be a nontrivial element of $\pi_{2\ell-1}(C)$.
 
\medskip

Now we begin with a formal proof of 
Lemma~\ref{l:confl-gadget-higher}.
Instead of modifying the proof of \cite{segal-skopenkov-spiez},
 we show how to reduce our assertion to theirs.
 Suppose there were a bad embedding $f$ of $S^k_a \sqcup S^k_b \sqcup \TG$ 
as in the lemma. 
Since the codimension of the image $f(S_a^k \sqcup S_b^k \sqcup \TG)$ 
is at least $2$, we can grow a $k$-dimensional finger from $f(S^k_a)$ 
towards $f(S^k_b)$ avoiding $f(\TG)$ until the finger touches $f(S^k_b)$ 
in a single point. This results in an embedding of $(S^k_a \vee S^k_b) 
\sqcup \TG$. For simplicity, we denote this modified embedding by $f$ as well.

 We observe that when pulling the finger, we can pull along a 
$(k+1)$-dimensional image of $B^{k+1}$ filling $f(S_a^k)$, 
and so the images are still linked or unlinked as
in the assumption of the lemma.

Next, consider the image $f(E)$ of the double lollipop in $C$. We modify $f$
as follows. We deformation retract the arc $f(c)$ to its midpoint $m$, pulling along
$\ell$-dimensional fingers from the two $\ell$-spheres $f(\Sigma^\ell_a)$ and 
$f(\Sigma^\ell_b)$, so that at the end of the deformation, the fingers touch in 
the single point $m$. This describes a continuous deformation of $f|_E$ that 
only changes $f|_E$ on the segment $c$ and in two small neigborhoods $U_a$ 
and $U_b$ of the endpoints of $c$ in the $\ell$-spheres (these neighborhoods 
provide the ``material'' for the fingers). We have to take care to pull along the parts 
of $B^{2\ell}$ attached to $U_a$ and $U_b$, respectively, i.e., we extend the deformation 
to a continuous deformation of $f$ on all of $X$ that changes $f$ only on a small neighborhood $V$ in $X$
of $c\cup U_a \cup U_b$. The whole deformation can be carried out so that the image of $V$ remains
in a small $\varepsilon$-neighborhood of the original image $f(E)$ throughout the deformation. 
Let $f'$ be the final modified map from $S_a^k \sqcup S_b^k \sqcup \TG$ into $S^d$ (note that 
we made no changes on the two $k$-spheres). The map $f'$ maps the ``bent stick'' 
$c$ of the double lollipop constantly to $m$ (in particular, it is not an embedding), and 
it induces a unique embedding $f''\colon X' \rightarrow C$ such that $f''$ agrees with $f'$ 
on the interior of $B^{2\ell}$ and $f'' \circ \kappa = f'$ on $E$, where $\kappa$ is the
map from Observation~\ref{obs:collapse-hemispheres}.
Moreover, the map 
$f\circ r=f\circ (\lambda \vee \lambda)\circ g\:S^{2\ell-1}
 \to C$ is deformed into the map 
$f'' \circ \kappa \circ (\lambda \vee \lambda) \circ g\:S^{2\ell-1}\to C$.
Thus, $f\circ r$ and $f''\circ \kappa \circ (\lambda \vee \lambda)\circ g$ 
define the same element of $\pi_{2\ell-1}(C)$. 
However, by the observation, the latter map is homotopic to 
$f''\circ g$. Thus, $f\circ r$ and $f''\circ g$ define the same 
element of $\pi_{2\ell-1}(C)$. But the former is trivial, 
as witnessed by $f\circ \overline{r}$, while the latter is 
not according to \cite{segal-skopenkov-spiez}---a contradiction. 
This completes the proof of (i).

\heading{Proof of (ii). }
For easier presentation, we describe an embedding $\overline f$ that
is not apriori PL; it is routine to replace it by a PL embedding.

Applying a suitable homeomorphism $\R^{2\ell+2}\to \R^{2\ell+2}$,
we may assume that $f(E)$ is actually $L^\ell\vee L^\ell$.
Let $\overline{L}^\ell$ denote the $\ell$-lollipop with
its $\ell$-sphere filled (i.e., $\overline L^\ell:=
B(-\frac12s_0,\frac12)\cup [0,s_0])$. It suffices to embed
$X$ in the $\delta$-neighborhood of $L^\ell\vee \overline L^\ell$
for $\delta>0$ arbitrarily small; actually, for notational
convenience, we will eventually get $4\delta$ instead of $\delta$.

Instead of specifying the embedding $\overline f\:\TG\to\R^{2\ell+2}$
directly,
we define a mapping $\tilde f\: S^\ell\times S^\ell \to \R^{2\ell+2}$
that coincides with $\lambda\vee\lambda$ on $S^\ell\vee S^\ell$
and maps the rest of $S^\ell\times S^\ell$ homeomorphically.
Then $\overline f$ can be given as (considering $E$
identified with $L^\ell\vee L^\ell$)
$$\overline f(z)=\alterdef{
z &\mbox{ for $z\in E$,}\\
\tilde f(z)&\mbox{ for $z\not\in E$.}}
$$

Writing a point of $S^\ell\times S^\ell$ as $(x,y)$,
we define $\tilde f$ using two auxiliary maps $u,v\:S^\ell\times [0,\infty)\to
\R^{\ell+1}$:
$$
\tilde f(x,y) := \Bigl(u(x,\dist(y,s_0)),v(y,\dist(x,s_0))\Bigr).
$$
For defining $u(x,t)$, we think of $t$ as time. For $t=0$,
the image $u(S^\ell,t)$ is the lollipop $L^\ell$, while
for all $t>0$ it is topologically a sphere, which looks 
almost like the lollipop; see Fig.~\ref{f:e4h-map-u}. 
\lipefig{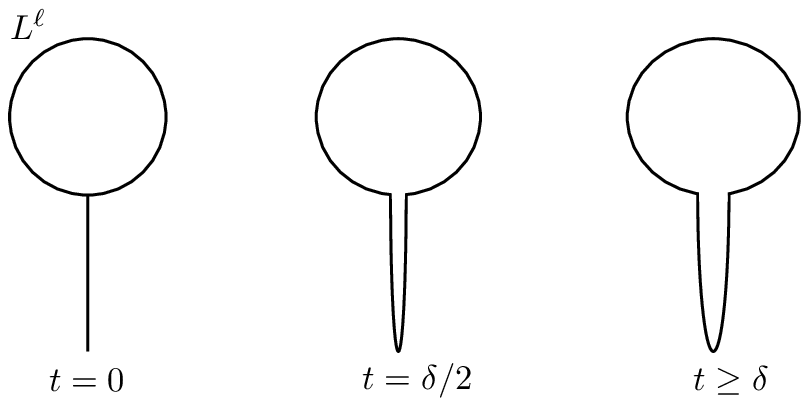}{The (images of the) mappings $u(*,t)$.}
Concretely, we set
$$
u(x,t):=\alterdef{ t x+(1-t)\lambda(x)&\mbox{ for } 0\le t\le\delta,\\
       \delta x+(1-\delta)\lambda(x)&\mbox{ for } t\ge\delta.}
$$
As for $v$, we let it coincide with $u$ for $t\le\delta$
(see Fig.~\ref{f:e4h-map-v}). For $t=2\delta$,
we set $v(x,2\delta):=(x_1,\delta x_2,\delta x_3,\ldots,\delta x_{\ell+1})$,
and for all $t\ge 3\delta$ we set
$v(x,t):= \delta(x-s_0)+s_0$ (i.e., the sphere is shrunk by
the factor of $\delta$ so that it still touches $s_0$).
 On the intervals $[\delta,2\delta]$
and $[2\delta,3\delta]$ we interpolate $v(x,t)$ linearly in $t$.
\lipefig{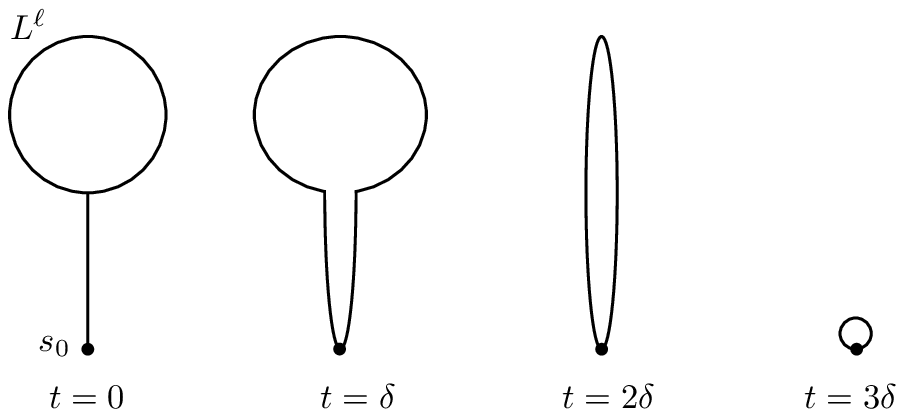}{The (images of the) mappings $v(*,t)$.}

The $\overline f$ defined in this way is clearly continuous
and  coincides with $\lambda\vee\lambda$ on $S^\ell\vee S^\ell$.
Next, we want to show $\overline f(x,y)\ne 
\overline f(x',y')$ whenever $(x,y)\ne(x',y')$ and none
of $x,x',y,y'$ equals $s_0$.
First we note that $u(x,t)\ne u(x',t')$
whenever $x\ne x'$ and $t,t'> 0$, and thus we may assume
$x=x'$, $y\ne y'$. Then we just use injectivity of $v(*,t)$
for every $t>0$.

It remains to check that the image of $\overline f$ lies
close to $\overline L^\ell\vee L^\ell$. The image $u(S^\ell,t)$ 
is $\delta$-close to $L^\ell$ for all $t$, and the image
$v(S^\ell,t)$ is $2\delta$-close to $s_0$ whenever $t\ge3\delta$.
Thus, whenever $\dist(x,s_0)\ge 3\delta$,
we have $\overline f(x,y)$ lying $3\delta$-close to $L^\ell\times \{s_0\}$.

Next, let us assume $\dist(x,s_0)\le 3\delta$. Then
$u(x,t)$ is $3\delta$-close to $s_0$ for all $t$,
and observing that $v(y,t)$ always lies $\delta$-close
to the filled lollipop $\overline L^\ell$, we conclude
 that $\overline f(x,y)$ is
$4\delta$-close to $\{s_0\}\times \overline L^\ell$.
\proofend

\subsection{The reduction}

Having introduced the clause gadget and the conflict gadget,
the rest of the reduction is almost the same as
in Section~\ref{ss:reduct}, and so we mainly point out the
(minor) differences. 

Given a 3-CNF formula $\varphi$, the simplicial complex
is pasted together from the gadgets exactly as in Section~\ref{ss:reduct};
we have $\dim K=\max(k,2\ell)=k$. For $\varphi$ unsatisfiable,
nonembeddability
of $K$ is shown using Lemmas~\ref{l:clause-gadget-higher}(i) and \ref{l:confl-gadget-higher}(i)
instead of Lemmas~\ref{l:cgadget}(i) and
\ref{l:confl-gadget}(i), but otherwise
in the same way as in Section~\ref{ss:reduct}.

Given a satisfiable formula $\varphi$, we again begin with
embedding the clause gadgets, this time using 
Lemma~\ref{l:clause-gadget-higher}(ii). 
%
For an opening
$\omega$ of a clause gadget $G_i$, we can again obtain 
a set $Q_\omega \subset \R^d$ with $Q_\omega\cap G_i=
\Sigma$, where $\Sigma=f(\bd\omega)$, this time homeomorphic
to $T^\eps\times B^{k}$ (where $T$
is the $(\ell+1)$-dimensional simplex bounded by $\Sigma$ and
$T^\eps$ is the part of it $\eps$-close to $\Sigma$).
Similarly we can build, for a non-witness opening
$\omega$, the set $Q_\omega^+$ homeomorphic
to $T\times B^{k}$.

Now we need to define within each $Q_\omega$ and $Q^+_\omega$
 the ``private pieces'' $Q_{\omega j}$ and $Q^+_{\omega j}$,
 $j=1,2,\ldots,t$,  within each $Q_\omega$ and $Q^+_\omega$, 
respectively. This time first we choose 
pairwise disjoint sets $B_1,\ldots,B_t\subset \bd B^{k}$,
each homeomorphic to $B^{\ell+1}$ (for this we need
$k\ge \ell+2$, which holds in our setting), we let $W_j$ be the
cone with base $B_j$ and apex in the center of $B^{k}$,
and we let $w_j$ be the boundary of $W_j$ (not including
the interior of the base $B_j$). We have $W_j$ homeomorphic
to $B^{\ell+2}$ and $w_j$ to $B^{\ell+1}$, and this allows us
to construct $Q_{\omega j}$ homeomorphic to a
$(2\ell+2)$-dimensional neighborhood of $\Sigma$,
and $Q^+_{\omega j}$ homeomorphic to a
$(2\ell+2)$-dimensional neighborhood of~$T$.

The rest of the embedding construction
can be copied from Section~\ref{ss:reduct} almost verbatim.
This concludes the proof of Theorem~\ref{t:4hard}.
\proofend

\section{Linkless embeddings}\label{s:linkless}

A PL embedding $f$ of a graph $G$ into $\R^3$ is called \emph{linkless} 
if the images of any two vertex-disjoint cycles in $G$ are unlinked, i.e.,
each of them bounds a PL disk that is disjoint from the other.

Robertson, Seymour, and Thomas~\cite{RobertsonSeymourThomas:SurveyLinklessEmbeddings-1993,RobertsonSeymourThomas:SachsLinklessEmbeddingConjecture-1995} 
showed, establishing a conjecture of Sachs,
 that a finite graph $G$ is linklessly embeddable in $\R^3$ if and only if $G$ does not contain 
one of the seven graphs in the so-called \emph{Petersen family} as a minor.
Moreover, they show (confirming a conjecture
by B\"ohme \cite{Boehme:SpatialRepresentationsGraphs-1990}) that 
every linklessly embeddable graph $G$
has even a \emph{panelled} embedding (also called a \emph{flat} embedding
in some sources) into $\R^3$,
i.e., a PL embedding such that for every cycle $C$ in $G$
 there exists a PL disk $D$ in $\R^3$ whose boundary equals $f(C)$ 
and that is otherwise disjoint from $f(G)$. 
It follows from the forbidden minor criterion that
linkless embeddability, as well as panelled embeddability,
can be tested in polynomial time 
(although the algorithm does not find  an embedding);
see \cite{RobertsonSeymourThomas:SurveyLinklessEmbeddings-1993,RobertsonSeymourThomas:SachsLinklessEmbeddingConjecture-1995}. 

The following lemma can be used to relate panelled embeddability 
to embeddability of $2$-dimensional complexes into $\R^3$.

\begin{lemma}[\textbf{B\"ohme}  \cite{Boehme:SpatialRepresentationsGraphs-1990}]\label{l:boehme}
Let $f$ be a panelled embedding of $G$ into $\R^3$, 
and let $C_1,\ldots, C_m$ be a family of cycles in $G$ 
any two of which are either disjoint or intersect in a path. 
Then there exist \textup{PL} disks $D_1,\ldots, D_m$ in $\R^3$ 
such that $\partial D_i =f(C_i)$ 
and the interiors of the $D_i$ are pairwise disjoint and disjoint from $f(G)$.
\end{lemma}

\begin{corol}
Let $K$ be a $2$-dimensional simplicial complex whose $1$-skeleton
does not have a minor from the Petersen family (and thus is
linklessly embeddable). Then $K$ embeds in $\R^3$.
\end{corol}
\begin{proof}
If $G$ is the $1$-skeleton of $K$,
 then the boundaries of the triangles in $K$ 
form a family of cycles as in Lemma~\ref{l:boehme}. 
Hence a panelled embedding of $G$ can be extended to an embedding of~$K$.
\end{proof}

We note that the general problem \EMBED{2}{3} 
can be rephrased as a \emph{partially panelled} embedding problem for graphs,
whose input is a graph $G$ and a family of triangles $C_1,C_2,\ldots,C_m$
in $G$, and the question is whether
$G$ admits a PL embedding in which each $C_i$ can be panelled.
This in itself does not tell us anything new about the
computational complexity of the problem, of course.

\subsection*{Acknowledgments}

We would like to thank 
Colin Rourke for explanations
concerning PL topology and for examples showing
the difference between PL embeddability and topological
embeddability mentioned in Section~\ref{s:pl-prelim}.
We also thank
Michael Joswig, Gil Kalai,
Frank Lutz, Alexander Nabutovsky,
and Robin Thomas
for kindly answering our questions. The second author
would also like to thank to Sergio Cabello for helpful
discussions regarding linear-time algorithms for \EMBED{2}{2}. 
Finally, we are grateful to two anonymous referees for careful
reading and valuable suggestions.

\bibliographystyle{alpha}
\bibliography{topj}

\appendix

\section{A decision algorithm for \boldmath
\EMBED 22 (sketch)}\label{s:22e}

Given a $2$-dimensional simplicial complex $K$,
we want to test whether it is embeddable in $\R^2$.
To this end, we can use a characterization of
$2$-dimensional simplicial complexes embeddable
in $\R^2$ due to Halin and Jung \cite{HalinJung}.
They give a list of seven small simplicial complexes,
denoted by $K_{\rm I}$ through $K_{\rm VII}$
and shown in Fig.~\ref{f:e4h-halin}, such that 
a $2$-dimensional simplicial complex $K$ is embeddable
in $\R^2$ iff it does not contain a subdivision
of some of $K_{\rm I}$--$K_{\rm VIII}$ as a subcomplex.
\lipefig{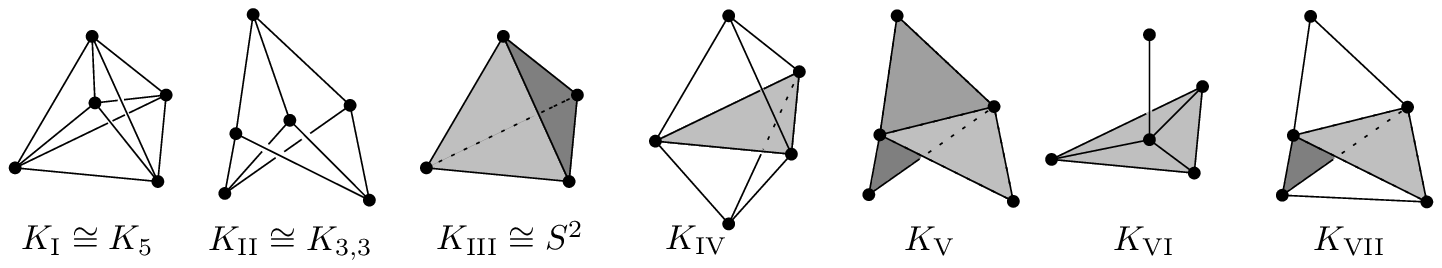}{The forbidden subcomplexes
$K_{\rm I}$--$K_{\rm VII}$.}

An inspection of $K_{\rm I}$--$K_{\rm VIII}$ reveals that
they are of three basic types:
\begin{enumerate}
\item[(a)] $K_{\rm III}$ is homeomorphic to an $S^2$.
\item[(b)] $K_{\rm VI}$ is a ``disk with a stick'',
i.e. a subdivision of a triangle with an edge attached
to a vertex in the middle.
\item[(c)] Each of the remaining five types contain
a subgraph isomorphic  to $K_{3,3}$ or $K_5$
in the $1$-skeleton or in the $1$-skeleton of their
first barycentric subdivision (see Fig.~\ref{f:e4h-halin1}
for the latter cases).
\end{enumerate}
(A slightly different characterization
of $2$-dimensional complexes embeddable in $S^2$
can be found in Marde\v{s}i\'c and Segal \cite{mardesic-segal};
it is somewhat less convenient for our purposes.)
\lipefig{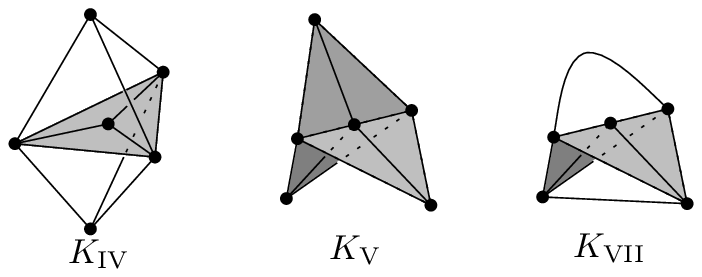}{$K_{\rm IV}$ and $K_{\rm V}$ contain
$K_{3,3}$, and 
$K_{VII}$ contains $K_5$.}

Thus, the following algorithm decides the embeddability
of a given $2$-dimensional complex $K$ into $\R^2$:

First, test if the $1$-skeleton of the first barycentric subdivision
of $K$ is a planar graph (this takes care of (c)),
and test if the link of each vertex of $K$
is either acyclic or consists of a single cycle
(this deals with (b)).

If both tests yield a positive answer, we know that
either $K$ embeds into the plane, or it contains a subdivision 
of $K_{\rm III}$, i.e., a 2-dimensional sphere. Moreover, 
since $K_{\rm V}$ is already excluded as a subcomplex, every 
edge of $K$ is incident to at most 2 triangles. This allows us 
to test in linear time if $K$ contains a homological cycle with 
$\Z_2$ coefficients: Consider the dual graph of $K$, whose vertices are the 
triangles of $K$, and two triangles are adjacent if they share 
an edge. Since every edge of $K$ is incident to at most
two triangles, every inclusion-minimal homological
cycle is an entire connected component of triangles in the
dual graph. 

Starting from an arbitrary triangle, we can find its component in linear time (using depth-first-search, say). If some triangle in
the component has a free edge with no other triangle incident to
it, we can discard the entire component and start again. If we can
discard all components of triangles in this way, we have also excluded
(a), and consequently $K$ is planar. Otherwise, we have found a 
collection $C$ of triangles in $K$ such that every edge is incident to either 
0 or 2 triangles of $C$, i.e., a homological cycle, which witnesses that $K$ does 
not embed in $\R^2$ (in fact, by the characterization of Halin and Jung, and
since (b) and (c) are already excluded, we know that
any minimial homological cycle we find must actually be 
a 2-dimensional sphere).

These tests can be performed in linear time. We remark that this only gives
a decision algorithm and does not actually construct an embedding. Daniel 
Kr\'al' (personal communication) has indepently devised a linear time algorithm, 
which moreover produces an embedding.

\section{The deleted product obstruction}\label{s:ha-we}

Here we recall the Haefliger--Weber theorem and some
related material.

Let $X$ be a topological space. The \emph{deleted product}
of $X$, which we denote by $X^2_\Delta$, 
is the following subspace of the Cartesian product
$X\times X$:
  $$
X^2_\Delta:=
X\times X\setminus\{(x,x):x\in X\}.
$$
An embedding $f\:X\to\R^d$ induces
a continuous map $\tilde f:X^2_\Delta\to S^{d-1}$ by 
$$
\tilde f(x,y):=\frac{f(x)-f(y)}{\|f(x)-f(y)\|}.
$$
Moreover, this map is \emph{equivariant},
which in this particular case means that
$\tilde f(y,x)=-\tilde f(x,y)$ for all $(x,y)\in X^2_\Delta$.
Thus, the existence of an equivariant map
of $X^2_\Delta$ in $S^{d-1}$ is necessary for
embeddability of $X$ into $\R^d$.

The Haefliger--Weber theorem shows that, surprisingly,
this necessary condition is also sufficient
if $X$ is (the polyhedron of) a $k$-dimensional simplicial
complex and $k$ lies in the metastable range
(mentioned after Theorem~\ref{t:4hard}), namely,
for $k\le (2d-3)/3$.
See Haefliger \cite{Haefliger} and Weber \cite{Weber67}
for original sources and,
e.g., Skopenkov \cite{skopenkov-survey}
for a modern overview, proof sketch, and extensions.
Thus, for example, for $(k,d)=(3,6)$, $(4,8)$,
$(5,9)$, etc., the embeddability
of a given $k$-dimensional simplicial complex into $\R^d$
is equivalent to
the existence of an equivariant
map of the space $\polyh{K}^2_\Delta$ into $S^{d-1}$.

For the case $d=2k$, the existence of such an equivariant map
can be decided using the \emph{first (equivariant) cohomological obstruction},
which in this particular case is equivalent 
to the Van Kampen obstruction mentioned in Appendix~\ref{s:vkamp}.

As was suggested by one of the referees, for the case $d=2k-1$,
the existence of an equivariant map is characterized
by (an equivariant version of)
a theorem of Steenrod \cite{steenrod-sq}, which involves
the first cohomological obstruction plus another apparently
computable invariant (a \emph{Steenrod square} operation in cohomology).
This might also yield even a polynomial-time decision algorithm,
but some computational issues still need to be checked.

For general $k$ and $d$, the existence of an equivariant
map $\polyh{K}^2_\Delta\to S^{d-1}$ can be
detected, e.g., via the \emph{extraordinary Van Kampen
obstruction} defined in Melikhov \cite{melikhov-vankampen}.
However, the computability of this obstruction remains
to be clarified. More generally, the computation of
the \emph{cohomotopy set} $[K,S^m]$ for a simplicial
complex $K$, or its equivariant version (which is what
we could use for the embeddability question),
might be possible by ``dualizing'' the methods for
computing higher homotopy groups mentioned in the
introduction.

\heading{Examples of incompleteness of the deleted product
condition. } It is known that the condition
$k\le (2d-3)/3$ defining the metastable
range in the Haefliger--Weber theorem is sharp,
in the following sense:
For each pair $(k,d)$ with  $d\ge 3$ and
$d\ge k> (2d-3)/3$, there exists a finite
simplicial complex $K$ of dimension (at most) $k$
that is not PL embeddable in $\R^d$, yet
such that an equivariant map $\polyh{K}^2_\Delta\to S^{d-1}$
exists. This was proved in several papers,
dealing with various values of $(k,d)$:
Marde\v{s}i\'c and Segal \cite{mardesic-segal-epsmappings}
($(d,d)$ for every $d\ge 4$), 
 Segal and Spie{\.z} \cite{segal-spiez}
(all values outside the metastable range with
finitely many exceptions), Freedman, Krushkal, Teichner \cite{FKT}
($k=2$, $d=4$), Segal, Skopenkov, and Spie{\.z} 
\cite{segal-skopenkov-spiez} (all remaining cases
except for $(k,d)=(2,3),(3,3)$), and
Gon\c{c}alves  and Skopenkov \cite{Go-Sko} 
($(k,d)=(2,3),(3,3)$).

\section{ PL embeddings versus topological embeddings}
\label{a:plvstop}

Let us say that \emph{TOP and PL embeddability coincide
for $(k,d)$} if every finite simplicial complex
of dimension at most $k$ that can be topologically
 embedded in $\R^d$ can also be PL embedded in $\R^d$.

The Alexander horned sphere   
and the other examples mentioned
in Section~\ref{s:pl-prelim}
 show that topological embeddings may exhibit behavior
that is impossible for PL embeddings. However, they do not 
clarify the relation of the class of $k$-dimensional complexes PL embeddable
in $\R^d$, i.e., the YES instances of \EMBED kd, to
the class of $k$-dimensional complexes topologically embeddable
in $\R^d$. If these two classes coincide,
we will say that \emph{TOP and PL embeddability coincide
for $(k,d)$}.

On the positive side, 
it is known that TOP and PL embeddability coincide for $(k,d)$
whenever $d-k \geq 3$ \cite{Bryant-PLapprox-codim3},
and also for $(k,d)=(2,3)$. The latter follows
from Theorem~5 of Bing \cite{Bing-3dmanifoldtriangulated},
which shows that
the image of a topological embedding of a $2$-dimensional
complex in $\R^3$ is homeomorphic to a polyhedron
PL embedded in $\R^3$, and from a result of
Papakyriakopoulos
\cite{Papakyriakopoulos:NewProofInvarianceHomologyGroups-1943}
(``Hauptvermutung'' for
$2$-dimensional polyhedra) that
any two $2$-dimensional polyhedra that are
homeomorphic are also PL homeomorphic.

However, TOP and PL embeddability do \emph{not} always 
coincide: 
There is an example of a $4$-dimensional
complex (namely, the suspension of the
Poincar\'e homology $3$-sphere)
 that embeds topologically, but not PL, into $\R^5$.
For this example we are indebted to Colin Rourke
(private communication); unfortunately, his proof,
although short, uses too
advanced concepts  to be reproduced here.   
It would be interesting to clarify  in general
for what $(k,d)$ TOP and PL embeddability coincide.

Let us also mention that simplicial complexes do not 
capture all ``finitary'' topological spaces whose embeddability
one might want to investigate. An interesting example
is Freedman's \emph{$E_8$ manifold} \cite{Freedman-top4manifolds},
which is a $4$-dimensional compact topological manifold that is not
homeomorphic to the polyhedron of any simplicial complex.

\section{The Van Kampen obstruction: An algorithmic perspective}\label{s:vkamp}

\heading{\boldmath The algorithm for deciding \EMBED k{2k}, $k\ne2$. }
Let $K$ be a $k$-dimensional finite simplicial complex.
We are going to define an object associated with $K$, 
the \emph{Van Kampen obstruction} $\vko_K$.

Let 
$$
P:=\{(\sigma,\tau): \sigma,\tau\in K, \sigma\cap\tau=\emptyset,
\dim\sigma=\dim\tau=k\}
$$
be the set of all ordered pairs of disjoint $k$-dimensional simplices
of $K$. From a rather pedestrian
point of view, which we mostly adopt here, 
$\vko_K$ is a subset of $\Z^P$, i.e., a set of integer vectors with
components indexed by $P$ (in the language of algebraic topology,
$\vko_K$ is an element of a certain cohomology group).
First we give a simple combinatorial description, and later on we
will offer a geometric interpretation, which will also explain
some of the terminology.

For defining $\vko_K$ we need to fix some (arbitrary) linear ordering
$\leq$ of the vertices of $K$ (although the final result does not depend
on it). We will write a $k$-dimensional simplex $\sigma\in K$
as $\sigma=[v_0,v_1,\ldots,v_k]$, meaning that $v_0$ through
$v_k$ are the vertices of $\sigma$ in increasing order under $\leq$.

First we define one particular vector $o_\gamma$ of  $\vko_K$.
The component corresponding to a pair $(\sigma,\tau)\in P$,
where $\sigma=[v_0,v_1,\ldots,v_k]$ and $\tau=[w_0,w_1,\ldots,w_k]$, is
$$
(o_\gamma)_{\sigma,\tau}:=
\alterdef{+1    &\mbox{ if $v_0<w_0<v_1<w_1<\cdots<v_k<w_k$},\\
         (-1)^k &\mbox{ if $w_0<v_0<w_1<v_1<\cdots<w_k<v_k$},\\
          0     &\mbox{ otherwise.}
}
$$

Next, we define a set $\Phi\subseteq \Z^P$, whose elements
we will call the \emph{finger move vectors}. The vectors in $\Phi$
correspond to pairs $(\omega,\nu)$ of simplices in $K$ such that one 
of $\omega,\nu$ has dimension $k$ and the other one dimension $k-1$.
Formally, we set
$$
Q:=\{(\omega,\nu):\omega,\nu\in K, \omega\cap\nu=\emptyset,
\dim\omega+\dim\nu=2k-1\},
$$
$$
\Phi:=\{\varphi^{\omega,\nu}: (\omega,\nu)\in Q\}.
$$
For defining the component $\varphi^{\omega,\nu}_{\sigma,\tau}$,
we again write $\sigma=[v_0,v_1,\ldots,v_k]$ and
$\tau=[w_0,w_1,\ldots,w_k]$.
We set
$$
\varphi^{\omega,\nu}_{\sigma,\tau} :=\alterdef{
(-1)^i&\mbox{ if $\nu=\tau$, $\omega=[v_0,v_1,\ldots,v_{i-1},
  v_{i+1},\ldots,v_k]$},\\
(-1)^{i+k}&\mbox{ if $\omega=\sigma$, $\nu=[w_0,w_1,\ldots,w_{i-1},
  w_{i+1},\ldots,w_k]$},\\
(-1)^{i+k}  &\mbox{ if $\omega=\tau$, $\nu=[v_0,v_1,\ldots,v_{i-1},
  v_{i+1},\ldots,v_k]$},\\
(-1)^i&\mbox{ if $\nu=\sigma$, $\omega=[w_0,w_1,\ldots,w_{i-1},
  w_{i+1},\ldots,w_k]$.}
}
$$
Thus, $\varphi^{\omega,\nu}_{\sigma,\tau}$ is nonzero only if
the $(k-1)$-dimensional simplex among
 $\omega,\nu$ is a facet of one of $\sigma,\tau$
and the $k$-dimensional simplex among $\omega,\nu$ equals
the other. 

Now, finally, we can define $\vko_K$ as the set of all vectors
that can be obtained from $o_\gamma$ by adding integer linear
combinations of vectors in $\Phi$; formally, 
$\vko_K:= o_\gamma+{\rm span}_\Z(\Phi)$. We say that $\vko_K$
\emph{vanishes} if it contains the zero vector.

We have the following result, based on ideas of \cite{vanKampen}
and proved independently by Shapiro~\cite{Shapiro-realiz2d}
and  Wu~\cite{Wu-realiz2d}, although their description of
the Van Kampen obstruction is different from ours:

\begin{theorem} \label{t:vkcomplete}
Let $K$ be a finite $k$-dimensional simplicial complex,
$k\ne 2$. Then $\vko_K$ vanishes if and only if $K$
can be PL embedded in $\R^{2k}$. 
\end{theorem}

This theorem provides a decision algorithm for \EMBED k{2k}, $k\ne 2$.
Indeed, it is clear from the above description that, given $K$, we can
set up $o_\gamma$
and $\Phi$ in polynomial time. The question of whether
$\vko_K$ vanishes amounts to testing whether $o_\gamma$ is an integer
linear combination of vectors in $\Phi$. This can be done
by bringing the matrix with the vectors of $\Phi$ as columns 
to the Smith normal form, for which several polynomial-time
algorithms are available in the literature.
The asymptotically fastest deterministic algorithm to date seems 
to be the one given in
\cite{Storjohann:NearOptimalAlgorithmsSmithNormalForm-1996};
fast randomized algorithms were given, 
for instance, in \cite{Giesbrecht:FastComputationSmithFormSparseIntegerMatrix-2001,DumasSaundersVillard:EfficientSparseIntegerMatrixSmithNormalForm-2001}. 
The latter seem to be particularly efficient in the case of sparse matrices.

\heading{Remarks. }
\begin{enumerate}
\item If $\vko_K$ does not vanish, then $K$ does not embed into $\R^{2k}$
even topologically \cite{Shapiro-realiz2d}, \cite{Wu-realiz2d}.
\item If $\vko_K$ vanishes, the Van Kampen approach also constructs
a PL embedding. However, nothing seems to be known about 
the complexity of the resulting embedding (the number of simplices
in a subdivision of $K$ on which the embedding is linear),
and it appears that a straightforward implementation of 
the construction may lead to an at least exponential complexity
in the worst case.
\item The case $k=2$ is indeed an exception; the point of
\cite{FKT} is to provide a $2$-dimensional simplicial complex $K$
that is not embeddable in $\R^4$, yet such that $\vko_K$ vanishes.
\item The various sign rules in the definition of the Van Kampen
obstruction are rather unpleasant; indeed, as  pointed out by Melikhov 
\cite{melikhov-vankampen}, a number of papers on the van Kampen 
obstruction (including  \cite{FKT}) contain a sign error. 
It would be much easier to work over $\Z_2$ instead of over
the integers (i.e., to reduce all components of the vectors
modulo $2$); 
then instead of using the Smith normal form, we could simply
solve a system of linear equations. However, an analog of Theorem~\ref{t:vkcomplete} for the Van Kampen obstruction reduced modulo~$2$ 
is valid only for $k=1$ (where better planarity algorithms exist
anyway), while for all $k\ge 3$ Melikhov 
\cite{melikhov-vankampen} provides examples of $k$-dimensional
complexes that are not embeddable in $\R^{2k}$ but whose
Van Kampen obstruction modulo~2 vanishes.
\item The Van Kampen obstruction can also be defined for
embedding of $k$-dimensional complexes into $\R^d$, $k\le d\le 2k$.
The definition is similar to the one given above but
formally more complicated (especially concerning the various signs).
We suspect (although we have no examples or references at present)
that the obstruction is incomplete whenever $d<2k$; that
is, that there are nonembeddable complexes for which the obstruction
vanishes. However,
it is still true that non-vanishing Van Kampen obstruction implies
non-embeddability, and so we have a potentially useful tool
for excluding many nonembeddable complexes. 
\item From the point of view of algebraic topology,
$\vko_K$ can be regarded as the first obstruction in integral cohomology
to the existence of an equivariant map of the deleted product
of $K$ into $S^{2k-1}$ (see Appendix~\ref{s:ha-we}). 
As such it is an element of a certain
equivariant cohomology group. See, e.g., \cite{melikhov-vankampen}.
\end{enumerate}

\heading{A geometric view of $\vko_K$. }
Let $f$ be a generic linear mapping of a finite $k$-dimensional
simplicial complex into $\R^{2k}$. Let us consider  a pair of disjoint
$k$-dimensional simplices $(\sigma,\tau)\in P$,
$\sigma=[v_0,v_1,\ldots,v_k]$, $\tau=[w_0,w_1,\ldots,w_k]$.
The images $f(\sigma)$ and $f(\tau)$ are $k$-dimensional
simplices in $\R^{2k}$, and because of genericity they
either are disjoint or intersect in a single point.
Moreover, again by genericity, the union of their vertex sets
is in general position and thus
\begin{eqnarray*}
B&:=&\Bigl(f(v_1)-f(v_0),f(v_2)-f(v_0),\ldots,f(v_k)-f(v_0),\\
&&\ \ \ \ f(w_1)-f(w_0),f(w_2)-f(w_0),\ldots,f(w_k)-f(w_0)\Bigr)    
\end{eqnarray*}
is an (ordered) basis of $\R^{2k}$.
We define the \emph{intersection number} of $f(\sigma)$ and $f(\tau)$
as
$$
f(\sigma)\cdot f(\tau):=
\alterdef{+1&\mbox{ if $f(\sigma)\cap f(\tau)\ne\emptyset$ and
$B$ is positively oriented},\\
          -1&\mbox{ if $f(\sigma)\cap f(\tau)\ne\emptyset$ and
$B$ is negatively oriented},\\
0& \mbox{ otherwise.}
}
$$
Considering how permuting a basis influences its orientation,
one can check that $f(\tau)\cdot f(\sigma)=(-1)^k f(\sigma)\cdot f(\tau)$.
Then we define a vector $o_f\in\Z^P$ by
\begin{equation}\label{e:def-o_f}
(o_f)_{\sigma,\tau} := (-1)^k f(\sigma)\cdot f(\tau).
\end{equation}

For reasons indicated later on, for every $f$ we have $o_f\in\vko_K$.
The $o_\gamma$ above corresponds to a particular 
choice of $f$. Namely, for $t\in\R$, let 
$$
\gamma(t):= (t,t^2,\ldots,t^{2k})\in\R^{2k};
$$
the set $\{\gamma(t):t\in\R\}$ is known as the
\emph{moment curve} in $\R^{2k}$. 
Let $V(K)=\{u_1,u_2,\ldots,u_n\}$ be the list of all
vertices of $K$ (ordered according to $\leq$). 

\begin{lemma} 
Let $f$ be the linear mapping of $K$ into $\R^{2k}$
given on the vertex set by $f(u_i):=\gamma(i)$. Then $f$ is generic
and we have $o_f=(-1)^{k(k-1)/2} o_\gamma$,
with $o_\gamma$ defined above Theorem~\ref{t:vkcomplete}.\footnote{We note that if $o_f\in \vko_K$ then $-o_f\in \vko_K$ as well,
since we can compose $f$ with a linear orientation-reversing
map $\R^{2k}\to\R^{2k}$. Hence the $(-1)^{k(k-1)/2}$
factor does not matter for the definition.}
\end{lemma}

\begin{proof}[Proof (sketch)] It is well known that every at most
$2k+1$ points on the moment curve in $\R^{2k}$ are affinely independent,
and thus the mapping is generic. Considering disjoint $k$-dimensional
simplices $\sigma,\tau\in K$ and
using \emph{Gale's evenness condition} \cite[Theorem~0.7]{Ziegler:Polytopes}, 
it is easy to see that $f(\sigma) \cap f(\tau) =\emptyset$ 
unless  the vertices of $\sigma$ and of $\tau$ alternate along 
the moment curve, in which case $f(\sigma)$ and $f(\tau)$
intersect (in a single point). 

It remains to determine the sign of $f(\sigma)\cdot f(\tau)$,
which equals the sign of the determinant of the matrix
(written as a list of columns)
$$
[\gamma(i_1){-}\gamma(i_0), \gamma(i_2){-}\gamma(i_0),\ldots, \gamma(i_k){-}
\gamma(i_0),
\gamma(j_1){-}\gamma(j_0), \gamma(j_2){-}\gamma(j_0),\ldots,
 \gamma(j_k){-}\gamma(j_0)]
$$
for some $i_0<j_0<i_i<j_1<\cdots<i_k<j_k$. 
If we permute the columns to the ordering
$\gamma(i_1)-\gamma(i_0)$, $\gamma(j_1)-\gamma(j_0)$,
$\gamma(i_2)-\gamma(i_0)$,\ldots, $\gamma(j_k)-\gamma(j_0)$,
the sign is multiplied by $(-1)^{k\choose 2}$.
A continuity argument and a calculation with Vandermonde determinants
show that for any choice of reals $x_0<y_0<x_1<y_1<\cdots<x_k<y_k$
the sign of $\det[\gamma(x_1)-\gamma(x_0),\gamma(y_1)-\gamma(y_0),
\ldots,\gamma(x_k)-\gamma(x_0),\gamma(y_k)-\gamma(y_0)]$ is
always $+1$.
\end{proof}

\medskip

The lemma just proved gives a geometric meaning to the
definition of $o_\gamma$.\footnote{Interestingly,
interpreting $\vko_K$ cohomologically
and expressing it as a cup product of a one-dimensional cohomology
class, one naturally arrives at the same element $o_\gamma\in\vko_K$.}
 It remains to explain where the
finger move vectors come from. We will also use this occasion
for outlining the main ideas in the proof of Theorem~\ref{t:vkcomplete}
(see \cite{FKT} for an insightful presentation of the proof).

To this end, we first extend the definition of $o_f$ to
a generic \emph{PL mapping} $f$ of $K$ into $\R^{2k}$.
Let $K'$ be a subdivision of $K$ such that $f$ is linear
on $K$. Then we extend the definition of the
intersection number $f(\sigma)\cdot f(\tau)$ by
$$
f(\sigma)\cdot f(\tau)=\sum_{\sigma'\lhd\sigma,\tau'\lhd\tau} f(\sigma')\cdot
f(\tau'),
$$
where the sum is over all pairs $(\sigma',\tau')$ such that
$\sigma'$ is a $k$-dimensional simplex of $K'$
contained in $\sigma$ and $\tau'$ is a $k$-dimensional simplex of $K'$
contained in $\tau$.\footnote{The definition of $f(\sigma)\cdot
f(\tau)$ for a linear mapping assumes a fixed linear ordering
of the vertices of $K$, which we do not have for the subdivision $K'$.
Here it helps to look at $K$ and $K'$ geometrically, assuming
that we deal with a geometric realization of $K$ in some $\R^m$.
Each $k$-dimensional simplex $\sigma=[v_0,v_1,\ldots,v_k]\in K$ 
spans a $k$-dimensional affine subspace $U$ of $\R^m$, and
each $\sigma'\lhd\sigma$ spans the same subspace.
We can thus choose a linear ordering $[v'_0,v'_1,\ldots,v'_k]$
of the vertices of $\sigma'$ such that the bases
$(v_1-v_0,\ldots,v_k-v_0)$ and $(v'_1-v'_0,\ldots,v'_k-v'_0)$
of $U$ have the same orientation. This does not necessarily
lead to a linear ordering of the whole vertex set of $K'$,
but the above definition of $f(\sigma)\cdot f(\tau)$ needs
only that  the vertex sets  $\sigma$ and $\tau$ are linearly
ordered, and  so we can choose a linear ordering
for each $\sigma'\in K'$ separately.}
Then $o_f$ is again defined by (\ref{e:def-o_f}).

Van Kampen's approach to the proof of Theorem~\ref{t:vkcomplete}
is based on the following geometric claim:
Given any two generic PL maps $f,g$ of $K$ into $\R^{2k}$,
there is a finite sequence $f_0=f,f_1,f_2,\ldots,f_t=g$
of generic PL maps of $K$ into $\R^{2k}$, such that
$f_{i+1}$ differs from $f_{i}$ only by a suitably defined
\emph{elementary move} (somewhat analogous to the Reidemeister
moves in knot theory). Then he analyzes the effect of the various
kinds of elementary moves on $o_f$; it turns out that the only
kind of move that can possibly affect $o_f$ is
a \emph{finger move}. 

To describe a finger move, we consider some generic PL mapping $f$
and two disjoint simplices $\omega,\nu\in K$,
$\dim\omega=k$, $\dim\nu=k-1$. The finger move corresponding
to $\omega$ and $\nu$ modifies $f$ by pulling a thin ``finger''
from $f(\omega)$ that wraps around $f(\nu)$; see Fig.~\ref{f:e4h-fingerm2d}
for an example with $k=1$. The finger can be chosen so  
that the modified $f$ is again a generic PL mapping.
\lipefig{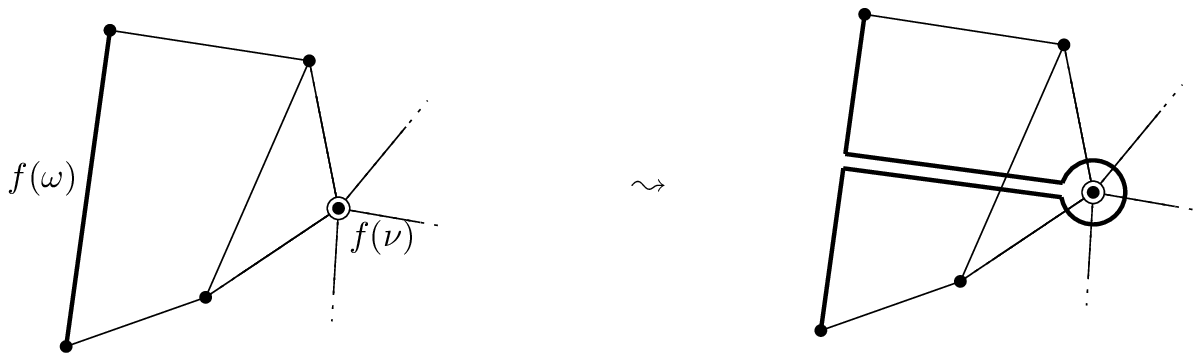}{A finger move.}

A careful analysis of the signs of the newly introduced intersections
shows that the finger move changes $o_f$  by 
one of the vectors $\varphi^{\omega,\nu}$, $-\varphi^{\omega,\nu}$,
which explains the definition of $\varphi^{\omega,\nu}$ given
above. 
Since any generic PL mapping of $K$ into $\R^{2k}$ can be transformed
into any other by a sequence of elementary moves (this statement we
have not proved, or even properly defined, of course, but it is not
difficult), it follows
that $o_f\in \vko_K$ for all $f$. In particular,
if $f$ is an embedding, we obviously have $o_f=0$,
and thus $\vko_K$ vanishes for embeddable $K$. 
This yields one implication in Theorem~\ref{t:vkcomplete}.

For the reverse (and harder) implication, we suppose that $\vko_K$
vanishes, which means that $o_\gamma$ can be expressed as an
integer linear combination of finger move vectors. 
This means that the generic linear mapping of $K$ into $\R^{2k}$
corresponding to $o_\gamma$ can be transformed by a sequence of finger
moves to a generic PL mapping $g$ with $o_g=0$. 
As is illustrated in Fig.~\ref{f:e4h-k4}, such a
$g$ is typically not yet an embedding, since the images of
non-disjoint simplices may intersect, and also the images of disjoint
simplices may have intersections---we only know that the number
of intersections always combines to $0$ algebraically
(taking the signs into account). One then needs a way of removing
these two kinds of intersections. There are three kinds of
elementary moves (different from finger moves) that
allow one to remove such intersections; one of them is
well known in topology as the \emph{Whitney trick},
and the others are similar. They are guaranteed
to work only for $k\ge 3$. (The graph case $k=1$
has to be treated separately and it works for a different reason.\footnote{For
graphs this step follows from the \emph{Hanani--Tutte} theorem
\cite{Hanani:UnplattbareKurven-1934,Tutte:TowardATheoryOfCrossingNumbers-1970},
which asserts that a graph that can be drawn with every two vertex-disjoint
edges intersecting an even number of times is planar.
In our language, this says precisely that if the Van Kampen obstruction
modulo~2 vanishes, then the graph is embeddable.})

\lipefig{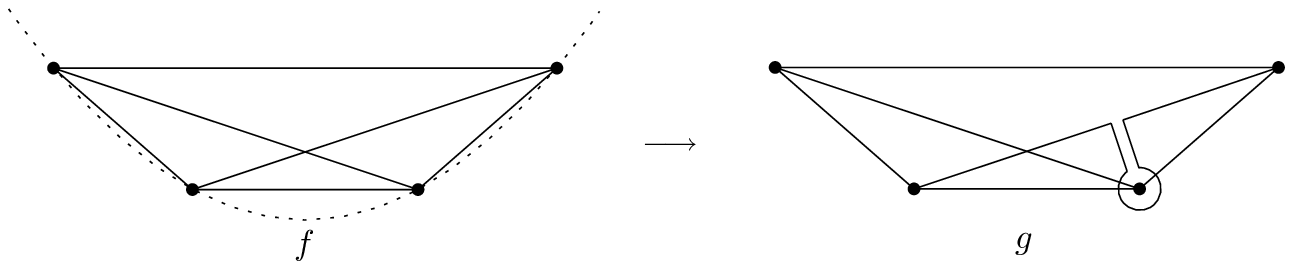}{A mapping $f$ of the graph $K_4$ into $\R^2$,
with vertices on the moment curve; a finger move
transforms it to a mapping $g$ with $o_g=0$.}

\end{document}